\documentclass[12pt]{article}

\usepackage{geometry}
\usepackage{yfonts}
\usepackage{subfigure}

\usepackage{amsthm}

\newcounter{mycou}

\usepackage{cancel}

\usepackage{pict2e}

\makeatletter
\newcommand{\CCOp}{\mathord{\mathpalette\nicoud@YESNO{\nicoud@path{\fillpath}}}}
\newcommand{\nicoud@YESNO}[2]{%
  \begingroup
  \settoheight{\unitlength}{$#1X$}%
  \begin{picture}(0.7,1)
  \linethickness{\variable@rule{#1}}%
  \roundcap\roundjoin
  \nicoud@path{\strokepath}
  #2
  \Line(0.35,-0.35)(0.35,1.08)
  \Line(0.55,-0.35)(0.55,1.08)
  \end{picture}%
  \endgroup
}
\newcommand{\nicoud@path}[1]{%
  \moveto(0.5,0.9)
  \lineto(0.5,1)\lineto(0.6,1)\lineto(0.6,0.9)
  \closepath
  \moveto(0.3,0.9)
  \lineto(0.3,1)\lineto(0.4,1)\lineto(0.4,0.9)
  \closepath
  \moveto(0.5,-0.27)
  \lineto(0.5,-0.17)\lineto(0.6,-0.17)\lineto(0.6,-0.27)
  \closepath
  \moveto(0.3,-0.27)
  \lineto(0.3,-0.17)\lineto(0.4,-0.17)\lineto(0.4,-0.27)
  \closepath
  #1
}
\newcommand{\variable@rule}[1]{%
  \fontdimen8  
  \ifx#1\displaystyle\textfont3\else
    \ifx#1\textstyle\textfont3\else
      \ifx#1\scriptstyle\scriptfont3\else
        \scriptscriptfont3\relax
  \fi\fi\fi
}
\makeatletter

\makeatletter
\newcommand{\subalign}[1]{%
  \vcenter{%
    \Let@ \restore@math@cr \default@tag
    \baselineskip\fontdimen10 \scriptfont\tw@
    \advance\baselineskip\fontdimen12 \scriptfont\tw@
    \lineskip\thr@@\fontdimen8 \scriptfont\thr@@
    \lineskiplimit\lineskip
    \ialign{\hfil$\m@th\scriptstyle##$&$\m@th\scriptstyle{}##$\hfil\crcr
      #1\crcr
    }%
  }%
}
\makeatother

 \newtheoremstyle{theoremdd}
  {0}
  {0}
  {\itshape}
  {0pt}
  {\bfseries}
  {}
  { }
  {\thmname{#1}\thmnumber{ #2}\textnormal{\thmnote{ (#3)}}}

\theoremstyle{theoremdd}

\usepackage{setspace,multirow}

\makeatletter
\newcommand*{\rom}[1]{\expandafter\@slowromancap\romannumeral #1@}
\makeatother

\usepackage{lmodern,enumerate,rotating,anyfontsize}
\usepackage{amsmath,amsfonts,amssymb,mathtools,blkarray,soul,etoolbox}
\usepackage{tikz-cd}
\usepackage{tikz-network}

\makeatletter
\newcommand{\mylabel}[2]{#2\def\@currentlabel{#2}\label{#1}}
\makeatother

\usetikzlibrary{cd}
\usetikzlibrary{shapes,shapes.geometric,arrows,fit,calc,positioning,automata}
\usetikzlibrary {circuits.logic.IEC}

\usepackage{stmaryrd}
\SetSymbolFont{stmry}{bold}{U}{stmry}{m}{n}

\newcommand{\blue}{\color{blue}}
\definecolor{green}{rgb}{0.1,0.7,0.1}

\newcommand{\cyan}{\color{cyan}}

\definecolor{Pantone540}{cmyk}{1 0.55 0 0.55}
\definecolor{Pantone301}{cmyk}{1 0.45 0 0.18}
\definecolor{Pantone285}{cmyk}{0.89 0.43 0 0}
\definecolor{Pantone542}{cmyk}{0.62 0.22 0 0.03}
\definecolor{Pantone283}{cmyk}{0.35 0.09 0 0}

\definecolor{Pantone383}{cmyk}{0.2 0 1 0.19}
\definecolor{Pantone158}{cmyk}{0 0.61 0.97 0}
\definecolor{Pantone7527}{cmyk}{0 0.02 0.06 0.07}

\DeclareMathOperator{\pre}{pre}

\DeclareMathOperator{\CC}{CC}

\DeclareMathOperator{\last}{last}
\DeclareMathOperator{\obs}{obs}

\DeclareMathOperator{\WEIGHT}{WT}

\allowdisplaybreaks

\newcommand{\Z}{\mathbb{Z}}
\newcommand{\R}{\mathbb{R}}
\newcommand{\N}{\mathbb{N}}
\newcommand{\Q}{\mathbb{Q}}
\newcommand{\LM}{\mathcal{L}}

\newcommand{\Acal}{\mathcal{A}}
\newcommand{\Acalf}{\mathcal{A}_{\mathsf{f}}}
\newcommand{\Acalnf}{\mathcal{A}_{\mathsf{nf}}}
\newcommand{\Rcal}{\mathcal{R}}
\newcommand{\dt}{\delta}

\newcommand{\ep}{\epsilon}

\newcommand{\Scal}{\mathcal{S}}

\newcommand{\Sig}{\Sigma}
\newcommand{\s}{\sigma}

\newcommand{\Mt}{\mathcal{M}}

\newcommand{\llb}{\llbracket}
\newcommand{\rrb}{\rrbracket}

\newcommand{\fsf}{\mathsf{f}}
\newcommand{\Ef}{{\cyan E_f}}

\newcommand{\ef}{{\cyan e_f}}

\newcommand{\Enf}{E_{\mathsf{nf}}}
\newcommand{\Eonf}{E_{o\mathsf{nf}}}
\newcommand{\Euonf}{E_{uo\mathsf{nf}}}
\newcommand{\Eof}{E_{o\mathsf{f}}}
\newcommand{\Euof}{E_{uo\mathsf{f}}}

\newcommand{\nf}{\mathsf{nf}}

\DeclareMathOperator{\init}{init}

\newtheorem{theorem}{Theorem}[section]
\newtheorem{definition}{Definition}

\newtheorem{remark}{Remark}

\usepackage{times,amssymb,amsfonts,mathtools,graphicx,tikz,algorithm,algorithmic,url,hhline,booktabs}
\usetikzlibrary{automata,arrows,positioning}

\newcounter{couRTA}
\usepackage{aliascnt}

\theoremstyle{plain}

\newaliascnt{lemma}{theorem}
\newtheorem{lemma}[lemma]{Lemma}
\aliascntresetthe{lemma}

\newaliascnt{proposition}{theorem}
\newtheorem{proposition}[proposition]{Proposition}
\aliascntresetthe{proposition}

\newaliascnt{corollary}{theorem}

\aliascntresetthe{corollary}

\newaliascnt{example}{theorem}
\newtheorem{example}[example]{Example}
\aliascntresetthe{example}

\usetikzlibrary{shadows}

\makeatletter
\let\NAT@parse\undefined
\makeatother
\usepackage[colorlinks,linkcolor=blue,anchorcolor=blue,citecolor=green,urlcolor=cyan,hyperindex]{hyperref}

\newcommand{\PSPACE}{\mathsf{PSPACE}}

\newcommand{\EXPTIME}{\mathsf{EXPTIME}}
\newcommand{\PTIME}{\mathsf{P}}
\newcommand{\EXPSPACE}{\mathsf{EXPSPACE}}

\newcommand{\NP}{\mathsf{NP}}
\newcommand{\coNP}{\mathsf{coNP}}

\usetikzlibrary{automata,arrows,positioning}
\usetikzlibrary{petri}

\tikzset{elliptic state/.style={draw, ellipse, thick, fill=gray!10}
}

\tikzset{rectangular state/.style={draw, rectangle, thick, fill=gray!10}
}

\tikzset{emptystate/.style={}
}

\tikzset{
    partial ellipse/.style args={#1:#2:#3}{
        insert path={+ (#1:#3) arc (#1:#2:#3)}
    }
}

\tikzset{
node distance=3cm, 
every state/.style={thick, fill=gray!10}, 
initial text=$ $, 
}

\newcommand{\preobserver}{\Acal_{\obs}^{\pre}}
\newcommand{\preobserverdt}{\dt_{\obs}^{\pre}}

\newcommand{\observer}{\Acal_{\obs}}
\newcommand{\observerdt}{\dt_{\obs}}



\usepackage{tcolorbox}
\usepackage{tabularx}
\usepackage{array}
\usepackage{colortbl}
\tcbuselibrary{skins}

\title{Resilience in labeled real-time automata}

\author{Kuize Zhang\\
{\small School of Mathematics and Statistic}\\
{\small Xi'an Jiaotong University, Xi'an, China}\\
{\small kuize.zhang@xjtu.edu.cn}
}

\begin{document}

\date{}

\maketitle

{\bf Abstract}
  In this paper, we characterize \emph{resilience} for a labeled real-time automaton (LRTA). An LRTA is resilient
  if whenever a faulty event occurs, after sufficiently many events occur, the LRTA returns to normalcy and the
  occurrence of the faulty event is not leaked. The notion of resilience reflects the ability of an LRTA recovering
  from a faulty behavior, and hence can model an intelligent agent. We formulate one definition of resilience
  for an LRTA and give verification algorithms for the definition based on two basic tools --- concurrent
  composition and observer --- proposed and computed in the author's previous papers [Yucan Yan, Kuize Zhang.
  Formulation of concurrent composition in labeled real-time automata
  with an application to fault diagnosis, 14th United Kingdom Automatic Control Council
  (UKACC) International Conference on Control, April, 10–12, 2024, Winchester, United Kingdom, 281--286] and
  [Kuize Zhang. State-based opacity of labeled real-time automata, Theoretical Computer Science,
  987 (2024), 114373 (1--20).].

{\bf Keywords}
resilience, labeled real-time automaton, verification, concurrent composition, observer

\tableofcontents

\section{Introduction}

\subsection{Background}

A partially-observed (aka labeled) dynamical system is \emph{resilient} if whenever
a faulty event occurs, the system will return to normalcy after sufficiently long
time and the occurrence of the faulty event is not leaked. A resilient system can be
regarded to be intelligent, because it is somehow not affected by external attack that
may cause fault.

The importance of resilience not only lies in that it is itself important,
but also lies in that it has relations with many important properties.
For example, in labeled finite-state automata (LFSAs), resilience has strong relations with the property --- diagnosability \cite{Sampath1995DiagnosabilityDES}, where the latter means that the occurrence of a faulty event can be
detected after sufficiently many events occur by observing the generated label sequence.
The diagnosability is not only important in computer science and automatic control \cite{Haar2017ComplexityActiveDiagnosis,Tripakis2002DiagnosisTimedAutomata,Cassez2012ComplexityCodiagnosability,Zhang2024DiagnosisDpAutomata} but also important in artificial intelligence \cite{Schumann2008Diagnosability_AAAI,Ibrahim2017DiagnosabilityDES_AAAI,Bittner2022DiagnosabilityFairTranSystem}, because an intelligent system should have the ability 
of detecting occurrences of its faults. From this point of view, resilience is somehow opposite to diagnosability, because resilience shows insensitivity to fault.
Resilience also has strong relations with opacity \cite{Mazare2004Opacity,Bryans2005OpacityPetriNet,Cassez2009DynamicOpcity,Saboori2007CurrentStateOpacity}, where the latter means
a visit of a secret state is not leaked by observing generated labels.
Opacity has important applications to cyber-security and cyber-privacy problems, e.g.,
chemical engineering and voting systems \cite{Bryans2008OpacityTransitionSystems},
encryption using pseudo random generators and tracking of mobile agents in sensor networks \cite{Saboori2010PhDThesisOpacity}, 
etc.

Although the notion of resilience is important,
there are very few formal definitions of resilience in automata. 
  In \cite{Fabre2022ResilienceDES}, several definitions of resilience were formulated
for LFSAs and their verification problems were proven to be $\PSPACE$-complete; in addition,
relations between resilience and diagnosability and between resilience and opacity were also revealed. In \cite{Akshay2021ResilienceTimedSystems}, two definitions of timed $K$-$\forall$-resilience and timed $K$-$\exists$-resilience were formulated for (fully-observed) timed automata, and their decision problems were proven to be undecidable and belong to $\EXPSPACE$, respectively; for integer reset timed automata (a subclass of timed automata), the decision problem for timed $K$-$\forall$-resilience was proven to be $\EXPSPACE$-complete. It was also proven that deciding untimed $K$-$\forall$-resilience and untimed $K$-$\exists$-resilience are $\EXPSPACE$-complete and $\PSPACE$-complete in timed automata in \cite{Akshay2021ResilienceTimedSystems}.
Physically, the definitions of resilience in \cite{Fabre2022ResilienceDES,Akshay2021ResilienceTimedSystems} both require that whenever faulty events occur, the automata will
return to normalcy after sufficiently long time (although the normalcy was formulated in  different ways in these two papers), 
and additionally, the definitions in \cite{Fabre2022ResilienceDES} also require that the occurrences of faulty events are not leaked. 
Therefore, the definitions of resilience formulated in \cite{Fabre2022ResilienceDES} are more reasonable.

A \emph{labeled real-time automaton} (LRTA) is widely used to model real-time systems. Although the LRTAs are a subclass of labeled timed automata and have strictly lower expressive power than labeled timed automata, the former preserve many decidable properties of  LFSAs because of the computable two basic tools --- \emph{concurrent composition} \cite{Yan2024DiagnosisReal-TimeAutomata} and \emph{observer} \cite{Zhang2024OpacityReal-TimeAutomata}.
In this paper, we formulate and characterize a definition of resilience for LRTAs by the two tools.
The resilience studied in the current paper is inspired and extended from the definitions given in \cite{Fabre2022ResilienceDES}.

\subsection{Review on the two basic tools needed}

We first review the two basic tools --- concurrent composition and observer in LFSAs, and then
review their nontrivial extensions in LRTAs.

The concurrent-composition method in LFSAs 
proposed in \cite{Zhang2020DetPNFA} by 
characterizing the negation of strong detectability provides a unified and currently the most efficient method
to verifying 
all kinds of inference-based properties\footnote{One can use an observed label sequence to determine internal information.} such as strong versions of detectability,
diagnosability, and predictability, for arbitrary LFSAs, while previous verification methods 
(e.g., the detector \cite{Shu2007Detectability_DES}, the twin-plant \cite{Jiang2001PolyAlgorithmDiagnosabilityDES}, and the verifier \cite{Yoo2002DiagnosabiliyDESPTime,Genc2009PredictabilityDES})
only apply to one type of such properties and they depend on 
two assumptions of deadlock-freeness and divergence-freeness, see \cite{Zhang2023UnifiedFrame4DES} for details.
Apart from verifying inference-based properties, when used in combination with the
observer\footnote{Originally
proposed in \cite{RabinScott1959PowersetConstruction} for determinizing nondeterministic finite
automata with $\varepsilon$-transitions. The terminology ``observer'' dates back to 
\cite{Ozveren1990ObservabilityDES,Shu2007Detectability_DES}.}, the concurrent composition provides
a unified and currently the most efficient method to verifying all kinds of concealment-based properties\footnote{One cannot use an observed label sequence to determine internal information.} such as 
standard versions of state-based opacity and strong versions of state-based opacity in LFSAs
\cite{Zhang2023UnifiedFrame4DES,Han2023StrongCSOpacity_DES}.
Comparisons of time complexities of different verification methods for standard versions of state-based opacity
and strong versions of state-based opacity in LFSAs were shown in \cite[Table~4 and Table~5]{Zhang2023UnifiedFrame4DES}.
In addition, the concurrent-composition method was used in \cite{Miao2025StrongInitialFinalOpacity} to verify
several kinds of strong initial-and-final-state opacity in LFSAs.
As a summary, the concurrent-composition method has fundamentally changed the state-of-the-art
of discrete-event systems modeled by LFSAs \cite{Zhang2023UnifiedFrame4DES}.
The concurrent composition of two LFSAs is computable in $\PTIME$ \cite{Zhang2020DetPNFA} while the observer of an LFSA is computable in $\EXPTIME$ \cite{RabinScott1959PowersetConstruction}. Along this line of using the concurrent composition and the observer to uniformly 
characterize inference-based properties and concealment-based properties in LFSAs,
if the two tools could be extended to LRTAs, then plenty of decidable properties obtained in LFSAs in the past around 30 years can be extended to LRTAs.
Such nontrivial extensions have become feasible because we have extended the two 
tools to LRTAs.

The concurrent composition was nontrivially extended to
LRTAs \cite{Yan2024DiagnosisReal-TimeAutomata} and the observer was nontrivially 
extended to LRTAs \cite{Zhang2024OpacityReal-TimeAutomata},
where the concurrent-composition computation problem was proven to be $\NP$-complete 
\cite{Yan2024DiagnosisReal-TimeAutomata} and the observer was proven to be 
computable in $\mathsf{N}$2$\EXPTIME$ \cite{Zhang2024OpacityReal-TimeAutomata}.
In \cite{Yan2024DiagnosisReal-TimeAutomata}, a notion of diagnosability was formulated and verified based on the
concurrent composition, and the diagnosability verification problem
was proven to be $\coNP$-complete. In \cite{Zhang2024OpacityReal-TimeAutomata}, several kinds of state-based opacity were formulated and verified based on the observer in 
$\mathsf{N}$2$\EXPTIME$.

In \cite{Zhang2022DetWAMonoid,Zhang2023DESbook}, a new class of automata called
labeled weighted automata over monoids were proposed, 
the definition of concurrent composition was nontrivially
extended to
labeled weighted automata over monoids, and necessary and sufficient conditions for strong versions of
detectability
were given based on the concurrent composition; particularly, for such automata over the monoid $(\Q^k,+)$,
the concurrent-composition computation problem was proven to be $\NP$-complete, and
the strong detectability verification problem was proven to be $\coNP$-complete. Later, 
in \cite{Miao2023VerifyDetectabilityUWASelf-Composition}, the concurrent composition was used to verify 
two strong versions of detectability of labeled unambiguous weighted automata over rational semirings
in exponential time (but the authors wrongly claimed that their algorithms ran in polynomial time,
they actually only considered the contribution of states to computational complexity but neglected
the contribution of the weights which are rational numbers). 
The observer was extended to labeled weighted automata and particularly computed in 
labeled weighted automata over the monoid $(\Q^k,+)$, and was used to verify weak detectability.
From the complexities of computing the two basic tools in LRTAs and in labeled weighted automata over the monoid $(\Q^k,+)$,
one can see similarities between these two 
classes of automata, although they perform remarkably different behavior.

With these two basic tools, we can introduce the main results in the current paper.

\subsection{Contribution}

In this paper, we formulate one definition of resilience for an LRTA, which intuitively means whenever a faulty event occurs, the system will return to normalcy after sufficiently many events occur and
the occurrence of the faulty event is not leaked. Based on our concurrent composition and observer, we give a necessary and sufficient condition for the negation of the resilience. Then, we prove that the negation of the resilience can be verified in $\mathsf{N}$2$\EXPTIME$ in the size of an LRTA.

\subsection{Structure of the paper}

In \autoref{sec:prelim}, we show necessary knowledge needed to show our results, including notation, the definitions on LRTAs, the concurrent composition \cite{Yan2024DiagnosisReal-TimeAutomata} and the observer \cite{Zhang2024OpacityReal-TimeAutomata}. In \autoref{sec:mainresult},
we show the main results of the current paper. \autoref{sec:conc} is a short conclusion.

\section{Preliminaries}
\label{sec:prelim}

\subsection{Notation}

For an alphabet $\Sigma$, $\Sigma^*$ denotes the set of finite strings (including the empty string
$\epsilon $) over $\Sigma$. Elements of an alphabet are called \emph{letters}.
Denote $\Sigma^+:=\Sigma^* \setminus \left\{\epsilon \right \}$.
For a string $w\in\Sig^*$, $|w|$ denotes its length, that is, the number of
elements of $\Sigma$, counting repetitions, occurring in $w$.
For a string $s_1\dots s_n$, where $s_1,\dots,s_n$ are letters, $s_1=:\init(s_1\dots s_n)$,
$s_n=:\last(s_1\dots s_n)$. 
As usual, $\R,\Q,\Z,\N$ denote the sets of real numbers, rational numbers, integers, nonnegative integers, 
respectively.
$\R_{\ge0}$ and $\R_+$ denote the sets of nonnegative real numbers and positive real numbers, respectively.
$\Q_{\ge0},\Q_+,\Z_{\ge0},\Z_+$ have analogous meanings.
$\llb m,n \rrb$ denotes the set of integers no less than 
$m$ and no greater than $n$. 
$\subset$ denotes set inclusion and $\subsetneq$ denotes strict set inclusion.

\subsection{Labeled real-time automata}
LRTAs are a special type of labeled timed automata. In an LRTA, there is only
one clock and the clock is reset upon each transition's execution.

An LRTA is a septuple 
\begin{align}
  \Acal=(Q,E,Q_0,\Delta,\mu,\Sig,\ell),
\end{align}
where $Q$ is a nonempty finite set of \emph{states}, $E$ is a finite \emph{alphabet} of \emph{events},
$Q_0\subset Q$ is a nonempty set of \emph{initial states}, $\Delta\subset Q\times E\times Q$ is
the \emph{transition relation} and elements of $\Delta$ are called \emph{transitions},
$\mu$ assigns to each transition $(q,e,q')\in\Delta$ (also written as $q\xrightarrow[]{e}q'$)
a nonempty interval $\mu(e)_{qq'}$ of $\R_{\ge0}$ with left
endpoint and right endpoint being $a$ and $b$, respectively, where $a\in\Q_{\ge0}$,
$b\in\Q_{\ge0}\cup\{+\infty\}$\footnote{When
  $b=+\infty$, the possible intervals can only be of the form $[a,+\infty)$ or $(a,+\infty)$;
when $a=b$, the possible intervals can only be $[a,a]=\{a\}$.}, 
$a\le b$, $\Sig$ is an alphabet of \emph{labels}, and $\ell:E\to\Sig\cup\{\ep\}$ is the \emph{labeling
function}. Removing $\mu$, $\Acal$ degenerates to an LFSA.

When $\Acal$ enters state $q$, the next event to occur is some $e$ such that $(q,e,q') \in \Delta$ and it will occur with a delay $t \in \mu(e)_{qq'}$; if no such event exists, $q$ is a \emph{dead} state from which no further evolution is possible.
When event $e\in E$
occurs, the \emph{label} $\ell(e)$ of $e$ will be observed if $\ell(e)\ne\ep$, in this case 
$e$ is called \emph{observable}; 
while nothing will be observed if $\ell(e)=\ep$, in this case $e$ is called \emph{unobservable}.
A transition $(q,e,q')$ is called \emph{observable} (resp., \emph{unobservable})
if $e$ is observable (resp., unobservable).
We denote by $E_o$ and $E_{uo}$ the set of observable events and the set of unobservable events, respectively.
Labeling function $\ell$ is
extended to $E\times \R_{\ge0}$ as follows: $\ell((e,t))=(\ell(e),t)$ if $e\in E_o$, $\ell( (e,t))=\ep$
otherwise. Then $\ell$ is recursively extended to $E^*$ as
$\ell(e_1\dots e_n)=\ell(e_1)\dots\ell(e_n)$ 
and also to $(E\times\R_{\ge0})^*$ analogously.

A \emph{path} of $\Acal$ is defined by 
a sequence $q_0\xrightarrow[]{e_1}q_1
\xrightarrow[]{e_2}\cdots\xrightarrow[]{e_{n}}q_n$, where $n\in\Z_{\ge0}$, $(q_{i-1},e_i,q_i)\in\Delta$ for all $i\in
\llb 1,n\rrb$. Its \emph{length} is $n$.
When $n=0$, the path degenerates to a single state $q_0$.
A path is called a \emph{cycle} if its start state and terminal state coincide.
For two states $q$ and $q'$, $q'$ is called \emph{reachable from} $q$ if there is a path
from $q$ to $q'$ with positive length. A state $q$ is called \emph{reachable} if either $q\in Q_0$ or $q$ is reachable from some
initial state.
A \emph{run} of $\Acal$ is 
a sequence $q_0\xrightarrow[]{e_1/t_1}q_1
\xrightarrow[]{e_2/t_2}\cdots\xrightarrow[]{e_{n}/t_{n}}q_n=:\pi$, where $n\in\Z_{\ge0}$,
$(q_{i-1},e_i,q_i)\in\Delta$, $t_i\in\mu(e_i)_{q_{i-1}q_i}$ for all $i\in\llb 1,n \rrb$.
A prefix of run $\pi$ is a run $q_0\xrightarrow[]{e_1/t_1}q_1
\xrightarrow[]{e_2/t_2}\cdots\xrightarrow[]{e_{m}/t_{m}}q_m$, where $m\le n$.
We say the run $\pi$ is over path $q_0\xrightarrow[]{e_1}q_1
\xrightarrow[]{e_2}\cdots\xrightarrow[]{e_{n}}q_n$. Sometimes, we write a run or a path as
$q_0\to^* q_n$\footnote{Here the symbol $*$ is chosen as Kleene star.} for short if the intermediate states, events, and times are not needed to be written explicitly.
Denote $q_0=:\init(q_0\to^* q)$ and $q=:\last(q_0\to^* q)$.
The \emph{timed word} of run
$\pi$ is defined by 
\begin{align}\label{eqn5_Resilience_LRTA}
  \tau(\pi)=(e_1,t_1')(e_2,t_2')\dots(e_n,t_n'),
\end{align}
where $t_i'=\sum_{k=1}^{i}t_k$ for all 
$i\in\llb 1,n \rrb$. The \emph{length} of a path/run/timed word is the length of its event sequence.
The \emph{weight}/\emph{duration} $\WEIGHT_\pi$, also denoted $\WEIGHT(\pi)$, of run $\pi$ and the \emph{weight}/\emph{duration} $\WEIGHT_{\tau(\pi)}$, also denoted $\WEIGHT(\tau(\pi))$,
of timed word $\tau(\pi)$ are both defined by $t_n'$. 
A path or run is called \emph{unobservable} if $\ell(e_1\dots e_n)=\ep$, and called \emph{observable} otherwise.
A run is called \emph{instantaneous} if its weight is equal to $0$. For a \emph{dead} state $q\in Q$, 
add an unobservable transition $(q,u,q)$
with $\mu(u)_{qq}=[0,+\infty)$. This modification is reasonable
because whenever $\Acal$
transitions to such a state $q$, it will always stay there and no label will be generated, but time will
still elapse. 
The concatenation of two runs $\pi_1:=q_0\xrightarrow[]{e_1/t_1}q_1
\xrightarrow[]{e_2/t_2}\cdots\xrightarrow[]{e_{n}/t_{n}}q_n$ and $\pi_2:=
q_n\xrightarrow[]{e_{n+1}/t_{n+1}}\cdots
\xrightarrow[]{e_{n+m}/t_{n+m}}q_{n+m}$ is defined as
$\pi_1\pi_2:=q_0\xrightarrow[]{e_1/t_1}q_1
\xrightarrow[]{e_2/t_2}\cdots\xrightarrow[]{e_{n+m}/t_{n+m}}q_{n+m}$.

Consider a run $q\xrightarrow[]{s}q'=:\pi$, where $s\in(E\times \R_{\ge0})^*(E_o\times \R_{\ge0})(E_{uo}\times \{0\})^*$.
Rewrite $\pi$ as $q\xrightarrow[]{s_1}q_1 \xrightarrow[]{s_2} \cdots \xrightarrow[]{s_n}
q'$, where $s_1,\dots,s_{n-1}\in (E_{uo}\times \R_{\ge0})^*(E_o\times \R_{\ge0})$,
$s_n\in (E_{uo}\times \R_{\ge0})^*(E_o\times \R_{\ge0})(E_{uo}\times \{0\})^*$,
denote 
\begin{align}\label{eqn1_Resilience_LRTA}
  \ell(\nu(\pi)) := \ell(\tau( q\xrightarrow[]{s_1}q_1 )) 
  \ell(\tau( q_1\xrightarrow[]{s_2}q_2 )) \dots \ell(\tau( q_{n-1}\xrightarrow[]{s_n}q' )).
\end{align}

For a run $\pi$ starting from some initial state, $\ell(\tau(\pi))\in(\Sig\times\R_{\ge0})^*$ is called a
\emph{timed label sequence generated by $\Acal$}. In this case, we observe
$\ell(e_i)$ at time $t_i'$ if $e_i\in E_o$, observe nothing at time $t_i'$ if $e_i\in E_{uo}$, $i\in\llb 1,n \rrb$, where $t_i'$ are defined in \eqref{eqn5_Resilience_LRTA}.
More generally, a sequence $(\sigma_1,t_1)\dots(\sigma_n,t_n)=:\gamma$ in $(\Sig\times\R_{\ge0})^*$ is called a 
\emph{timed label sequence} if $t_1\le \cdots \le t_n$.
For two runs $\pi_1,\pi_2$ starting from $Q_0$, they are called \emph{observationally
equivalent} if $\ell(\tau(\pi_1)) = \ell(\tau(\pi_2))$, i.e., the observations to them are the same.
The \emph{weight} $\WEIGHT_{\gamma}$
of timed label sequence $\gamma$ is defined as $t_n$. Particularly, $\WEIGHT_{\ep}:=0$.
The \emph{length} of a timed label sequence is the length of its label sequence.
The \emph{timed language} $L(\Acal)$ generated by $\Acal$ is defined by the set of timed words of all 
runs of $\Acal$ starting from initial states; $\LM(\Acal)$ is the set of timed label sequences generated
by $\Acal$.
Particularly, denote $L_q(\Acal)$ as the set of timed words of all runs of $\Acal$
starting from $q$.

\subsection{State estimation} 

For $\Acal$, a subset $x\subset Q$ of states, and a sequence
$\gamma\in(\Sig\times \R_{\ge0})^+$, we define the \emph{current-state estimate} as 
\begin{equation}\label{CSE_RTautomata}
	\begin{split}
		\Mt(\Acal,\gamma|x):=\{q\in Q|&(\exists q_0\in x)(\exists n\in\Z_{+})(\exists m\in\N)\\
							&\left(\exists\text{ a run }\pi=q_0\xrightarrow[]{e_1/t_1}\cdots\xrightarrow[]{e_n/t_n}q_n
							\xrightarrow[]{e_{n+1}/0}\cdots\xrightarrow[]{e_{n+m}/0}q\right)\\
							&[(e_n\in E_o) \wedge (e_{n+1}\dots e_{n+m}\in (E_{uo})^*) \wedge
							\ell(\tau(\pi))=\gamma]\}.
	\end{split}
\end{equation}

Particularly for $\Acal$ and $x\subset Q$, we define the \emph{instantaneous-state estimate} as
\begin{equation}\label{ISE_RTautomata}
	\begin{split}
	  \Mt(\Acal,(\ep,0)|x):=x\cup
						  \{q\in Q|&(\exists q_0\in x)(\exists n\in\Z_{+})
						  (\exists\text{ a run }\pi=q_0\xrightarrow[]{e_1/0}\cdots\xrightarrow[]{e_n/0}q)\\
						  &[e_1\dots e_n\in  (E_{uo})^*]\}.
	\end{split}
\end{equation}

For all $\gamma\in(\Sig\times\R_{\ge0})^+$, $\Mt(\Acal,\gamma|Q_0)$ is also rewritten as $\Mt(\Acal,\gamma)$
for short.
Intuitively, for $\gamma=(\s_1,t_1)\dots(\s_n,t_n)\in(\Sig\times \R_{\ge0})^+$, $\Mt(\Acal,\gamma|x)$ 
denotes the set of states $\Acal$ can be in when $\gamma$ has just been generated by $\Acal$ since $\Acal$
started from some state of $x$. 
In order to fit the setting of current-state estimate, after the occurrence of the last observable event
$e_n$ (i.e., $e_n$ occurs at the current time), we only allow unobservable, instantaneous runs,
which is represented by $q_n\xrightarrow[]{e_{n+1}/0}\cdots\xrightarrow[]{e_{n+m}/0}q$ and $e_{n+1}\dots 
e_{n+m}\in (E_{uo})^*$. Particularly, $\Mt(\Acal,(\ep,0)|x)$ denotes the set of states $\Acal$ can be in at the instant 
when $\Acal$ just transitions to some state of $x$. Since there may exist instantaneous transitions, at the
instant, $\Acal$ may be in some state outside of $x$.

$\Mt(\Acal,(\ep,0)|Q')$ can be computed in $\PTIME$, while
$\Mt(\Acal,\gamma|Q')$ can be computed in $\NP$ if all $t_1,\dots,t_n$ belong to $\Q$
\cite[Theorem~3.1]{Zhang2024OpacityReal-TimeAutomata}.

\subsection{The notion of concurrent composition}

In this subsection we recall the tool --- concurrent composition \cite{Yan2024DiagnosisReal-TimeAutomata}. 
For two LFSAs, in their concurrent composition,
observable transitions are synchronized and
unobservable transitions interleave \cite{Zhang2020DetPNFA}. The concurrent composition of two LRTAs
\cite{Yan2024DiagnosisReal-TimeAutomata} is a natural and nontrivial generalization of the concurrent composition of two LFSAs.
The concurrent composition of two LFSAs can be computed in time polynomial in the sizes of the two LFSAs
\cite{Zhang2020DetPNFA}, but the concurrent-composition computation problem of two LRTAs is $\NP$-complete
\cite{Yan2024DiagnosisReal-TimeAutomata}. We will use a variant of the concurrent composition proposed
and computed in \cite{Yan2024DiagnosisReal-TimeAutomata}.

\begin{definition}\label{def:CCALRTA}
  For an LRTA $\mathcal{A}$, the concurrent composition of $\Acal$ and itself, called the 
  \emph{self-composition} of $\Acal$, denoted as $\CC(\Acal)$, is defined by 
  \begin{align}
	( Q', E_o', Q_{0}', \Delta ', \mu',\Sigma, \ell') ,
  \end{align}where
\begin{itemize}
\item $Q'= Q\times Q$;
\item $E'_{o}=\{(e_{1},e _{2})\in E_{o}\times E_{o}|\ell(e_{1})=
  \ell(e_{2})\}$;
\item $Q_{0}'=Q_{0}\times Q_{0}$ is the set of initial states;
\item $\Delta'\subset Q'\times E'\times Q'$ is the transition relation;
\item $\ell'(e_{1},e_{2})=\ell(e_{1})=\ell(e_{2})$ for all $(e_{1},e_{2})\in E_o'$.
\end{itemize}

For all states $( q_{1},q_{2}),( q_{3},q_{4}) \in Q'$ and events
$( e_{1},e_{2})\in E_o'$, $(( q_{1},q_{2}),( e_{1},e_{2}),( q_{3},q_{4}))\in\Delta '$ if and only if there are two runs:
\begin{subequations}\label{eqn_admissiblerun_RTautomata}
\begin{align}
  \pi _1 & :=q_1\to^* q_5 \xrightarrow[]{e_1/t_1} q_7 \to^* q_3, \label{eqn1_admissiblerun_RTautomata}\\
  \pi _2 & := q_2 \to^* q_6 \xrightarrow[]{e_2/t_2} q_8 \to^* q_4, \label{eqn2_admissiblerun_RTautomata}
\end{align}
\end{subequations}
where $q_5, q_6,q_7,q_8\in Q $, $t_1\in \mu(e_1)_{q_5q_7}$, $t_2\in \mu(e_2)_{q_6q_8}$,
$\WEIGHT_{\pi_1}= \WEIGHT_{\pi_2}$, 
all events except for $e_1$ and $e_2$ are unobservable,
$\WEIGHT_{q_7 \to^* q_3} = \WEIGHT_{q_8 \to^* q_4} = 0$.
Such $\pi_1$ and $\pi_2$ are called left and right \emph{admissible runs} of transition $((q_{1},q_{2}),(e_{1},e_{2}),
( q_{3},q_{4}))$. 
\end{definition}

An observable transition 
$({q}_{1},{q}_{2})\xrightarrow[]{({e}_{1},{e}_{2})}({q}_{3},{q}_{4})$ in $\CC(\mathcal{A})$
is interpreted as follows: 
at the beginning $\mathcal{A}$ is in state $q_1$ or $q_2$ and transition to state $q_3$ or $q_4$
after some common time delay when event $e_1$ or $e_2$ occurs. That is, in the two cases, the occurrences of the
two observable events $e_1$ and $e_2$ are synchronized. Therefore, after $e_1$ and $e_2$, we only consider 
instantaneous transitions. Positive weight for the observable transition
means that the transition from $q_1$ to $q_3$ can cost positive time. Zero weight for the observable transition means 
that the transition from $q_1$ to $q_3$ costs no time.

\begin{example}
\refstepcounter{couRTA}
  An LRTA $\mathcal{A}_{\thecouRTA\label{couRTA2}}$ is depicted in \autoref{fig1_det_RTautomata}. The reachable 
  part of self-composition $\CC(\Acal_{\ref{couRTA2}})$ is illustrated in \autoref{fig2_det_RTautomata}. 
  $(( q_{0},q_{0}),( e_{1},e_{2}), ( q_{3},q_{4}))$ is an observable transition,
  ${q}_{0}\xrightarrow[]{u/1}{q}_1\xrightarrow[]{e_1/2}q_3$ and ${q}_{0}\xrightarrow[]{u/0.9}{q}_2\xrightarrow[]
  {e_2/2.1}q_4$ are two of its admissible runs, where both of them produce the timed label sequence
  $(\sigma,3)$.

  \begin{figure}[!htbp] 
  \centering
	\begin{tikzpicture}
	[>=stealth',shorten >=1pt,thick,auto, node distance=2.5 cm, scale = 0.8, transform shape,
	->,>=stealth,inner sep=2pt, initial text = {}]

	\tikzstyle{emptynode}=[inner sep=0,outer sep=0]

	\node[initial, state, initial where = left] (q0) {$q_0$};
	\node[state] (q2) [below right of = q0] {$q_2$};
	\node[state] (q1) [above right of = q0] {$q_1$};
	\node[state] (q3) [right of = q1] {$q_3$};
	\node[state] (q4) [right of = q2] {$q_{4}$};

	\path [->]
	(q0) edge node [above, sloped] {$u/[1,2]$} (q1)
	(q0) edge node [above, sloped] {$u/[0.6,1]$} (q2)
	(q3) edge [loop right] node {$e_1/[0,1]$} (q3)
	(q4) edge [loop right] node {$e_2/[1,2]$} (q4)
	(q1) edge node [above, sloped] {$e_1/[1,2]$} (q3)
	(q2) edge node [above, sloped] {$e_2/[2.1,3]$} (q4)
	;

     \end{tikzpicture}
	 \caption{LRTA $\Acal_{\ref{couRTA2}}$, where $u$ is
	 unobservable, $e_1$ and $e_2$ are observable and $\ell(e_1)=\ell(e_2)=\sigma$.}
	 \label{fig1_det_RTautomata}
  \end{figure}
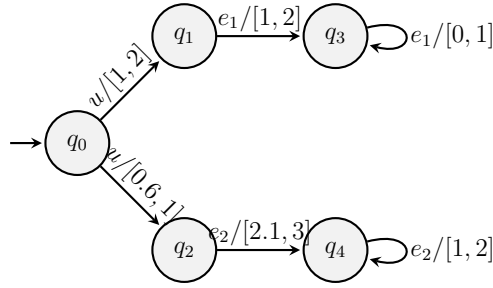
  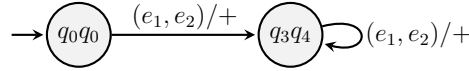
\begin{figure}[!htbp] 
  \centering
  	 \begin{tikzpicture}
	[>=stealth',shorten >=1pt,thick,auto, node distance=3.5 cm, scale = 0.8, transform shape,
	->,>=stealth,inner sep=2pt, initial text = {}]

	\tikzstyle{emptynode}=[inner sep=0,outer sep=0]

	\node[initial, state, initial where = left] (00) {$q_0q_0$};
	\node[state] (34) [right of = 00] {$q_3q_4$};
	
	\path [->]
	(00) edge node [sloped, above] {$(e_1,e_2)/+$} (34)
	(34) edge [loop right] node {$(e_1,e_2)/+$} (34)
	;

    \end{tikzpicture}
	\caption{Reachable part of self-composition $\CC(\mathcal{A}_{\ref{couRTA2}})$, where LRTA $\Acal_{\ref{couRTA2}}$ is shown in 
	\autoref{fig1_det_RTautomata}.}
	\label{fig2_det_RTautomata}
  \end{figure}
\end{example}

\begin{theorem}[\cite{Yan2024DiagnosisReal-TimeAutomata}]\label{thm1_RTautomata_diagnosis}
  Consider an LRTA $\Acal$, its self-composition $\CC(\Acal)$ computation problem is $\NP$-complete.
\end{theorem}

\subsection{The notion of observer}\label{subsec:observer}

In this subsection we recall the notion of \emph{observer} $\Acal_{\obs}$ \cite{Zhang2024OpacityReal-TimeAutomata} that
concatenates current-state estimates along timed label sequences generated by $\Acal$.

Before defining $\observer$, we need to define a notion of \emph{pre-observer} $\preobserver$.

\begin{definition}[\cite{Zhang2024OpacityReal-TimeAutomata}]\label{def_preobserver_RTautomata}
For an LRTA $\Acal$, we define its \emph{pre-observer} as a deterministic automaton 
\begin{align}\label{eqn_preobserver_RTautomata}
	\Acal_{\obs}^{\pre}=(X,\Sig\times \R_{\ge0},x_0,\dt_{\obs}^{\pre}),
\end{align}
where $X= 2^Q$ is the state set, $\Sig\times \R_{\ge0}$ the alphabet, 
$x_0=\Mt(\Acal,(\ep,0))\in X$
the unique initial state, $\preobserverdt\subset X\times (\Sig\times \R_{\ge0})\times X$ the transition relation.
For all $x,x'\in
X$ and $(\s,t)\in \Sig\times \R_{\ge0}$, $(x,(\s,t),x')\in\preobserverdt$ if and only if $x'=\Mt(\Acal,
(\s,t)|x)$.
\end{definition}

In \autoref{def_preobserver_RTautomata}, after $\preobserverdt$ is recursively
extended to $\preobserverdt\subset X\times (\Sig\times \R_{\ge0})^*\times X$, one has for all $x\in X$ and 
$(\s_1,t_1)\dots(\s_n,t_n)=:\gamma\in (\Sig\times \R_{\ge0})^{+}$, $(x_0,\gamma,x)\in\preobserverdt$ 
if and only if $\Mt(\Acal,\tau(\gamma))=x$, 
i.e., $x$ is the set of states that $\Acal$ can be in when timed label sequence $\tau(\gamma)$ has just been
generated.

Note that the alphabet $\Sig\times \R_{\ge0}$ is not finite, so one cannot compute
the whole $\Acal_{\obs}^{\pre}$. Next, we define observer $\observer$ as a computable sub-automaton of
$\preobserver$.

\begin{definition}[\cite{Zhang2024OpacityReal-TimeAutomata}]\label{def_observer_RTautomata}
	For an LRTA $\Acal$, consider its pre-observer \eqref{eqn_preobserver_RTautomata}, 
	we define its \emph{observer} as a finite automaton  
\begin{align}\label{eqn_observer_RTautomata}
	\observer=(X,\Sig_{\obs},x_0,\dt_{\obs}),
\end{align}
where $\Sig_{\obs}$ (resp., $\observerdt$) is a finite subset of $\Sig\times\Q_{\ge0}$ (resp., 
$\preobserverdt$), such that if there exists a transition
from $x\in X$ to $x'\in X$ in $\preobserverdt$ then at least one such transition belongs to
$\observerdt$.
\end{definition}

Note that an LRTA $\Acal$ may have multiple observers, because $\Sig_{\obs}$ may not be unique;
however, $X$ and $x_0$ must be unique. 
In \autoref{def_observer_RTautomata}, after $\dt_{\obs}$ is recursively
extended to $\dt_{\obs}\subset X\times (\Sig_{\obs})^*\times X$, one has for all $x\in X$ and 
$(\s_1,t_1)\dots(\s_n,t_n)=:\gamma\in (\Sig_{\obs})^{+}$, $(x_0,\gamma,x)\in\dt_{\obs}$ if and only if
$\Mt(\Acal,\tau(\gamma))=x$.


\begin{example}\label{exam5_RTautomata}
\refstepcounter{couRTA}
  Consider the LRTA $\Acal_{\thecouRTA\label{couRTA3}}$ shown in \autoref{fig1_RTautomata}.
	One of its observers is shown in \autoref{fig2_RTautomata}. 

	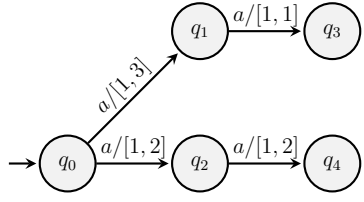
\begin{figure}[!h]
        \centering
	\begin{tikzpicture}
	[>=stealth',shorten >=1pt,thick,auto,node distance=2.5 cm, scale = 0.7, transform shape,
	->,>=stealth,inner sep=2pt]

	\tikzstyle{emptynode}=[inner sep=0,outer sep=0]

	\node[initial, state, initial where = left] (q0) {$q_0$};
	\node[state] (q2) [right of = q0] {$q_2$};
	\node[state] (q1) [above of = q2] {$q_1$};
	\node[state] (q3) [right of = q1] {$q_3$};
	\node[state] (q4) [right of = q2] {$q_{4}$};

	\path [->]
	(q0) edge node [above, sloped] {$a/[1,3]$} (q1)
	(q0) edge node [above, sloped] {$a/[1,2]$} (q2)
	(q1) edge node [above, sloped] {$a/[1,1]$} (q3)
	(q2) edge node [above, sloped] {$a/[1,2]$} (q4)
	;

     \end{tikzpicture}
	 \caption{An LRTA $\Acal_{\ref{couRTA3}}$, $\ell(a)=a$.}
	 \label{fig1_RTautomata}
	\end{figure}

	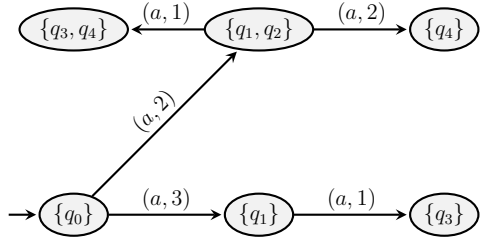
\begin{figure}[!h]
        \centering
	\begin{tikzpicture}
	[>=stealth',shorten >=1pt,thick,auto,node distance=3.5 cm, scale = 0.7, transform shape,
	->,>=stealth,inner sep=2pt]

	\tikzstyle{emptynode}=[inner sep=0,outer sep=0]

	\node[initial, elliptic state, initial where = left] (q0) {$\{q_0\}$};
    \node[elliptic state] (q1) [right of = q0] {$\{q_1\}$};
	\node[elliptic state] (q12) [above of = q1] {$\{q_1,q_2\}$};
	\node[elliptic state] (q3) [right of = q1] {$\{q_3\}$};
    \node[elliptic state] (q4) [above of = q3] {$\{q_4\}$};
	\node[elliptic state] (q34) [left of = q12] {$\{q_3,q_4\}$};

	\path [->]
	(q0) edge node [above, sloped] {$(a,2)$} (q12)
	(q0) edge node [above, sloped] {$(a,3)$} (q1)
	(q1) edge node [above, sloped] {$(a,1)$} (q3)
	(q12) edge node [above, sloped] {$(a,2)$} (q4)
	(q12) edge node [above, sloped] {$(a,1)$} (q34)
	;

     \end{tikzpicture}
	 \caption{One observer $\Acal_{\ref{couRTA3}\obs}$ of LRTA $\Acal_{\ref{couRTA3}}$ in \autoref{fig1_RTautomata}.}
	 \label{fig2_RTautomata}
	\end{figure}

\end{example}

\section{Main results}
\label{sec:mainresult}

\subsection{Formulation of a definition of resilience}

In order to formulate resilience, we specify a subset $\Ef\subset E$ of faulty events.
Then denote $\Enf:=E\setminus \Ef$.
A \emph{faulty run} is a run with at least one event in $\Ef$. A \emph{non-faulty run}
is a run with all events in $\Enf$.
We denote $\Acalf$ (resp., $\Acalnf$) as the subautomaton of $\Acal$
obtained by removing all transitions of $\Acal$ with events in $\Enf$ (resp., $\Ef$).
Denote $\Eonf:=E_o\cap \Enf$, $\Euonf:=E_{uo}\cap \Enf$,
$\Eof:=E_o\cap \Ef$, $\Euof:=E_{uo}\cap \Ef$.

By \autoref{def_preobserver_RTautomata}, the pre-observer of $\Acalnf$ is denoted by
\begin{align}\label{def_pre-observer_RTautomata_nf}
  \Acal_{\nf\obs}^{\pre}=(X,\Sig\times \R_{\ge0},x_{0\nf},\dt_{\nf\obs}^{\pre}),
\end{align}
the pre-observer of $\Acalf$ is denoted by
\begin{align}\label{def_pre-observer_RTautomata_f}
  \Acal_{\fsf\obs}^{\pre}=(X,\Sig\times \R_{\ge0},x_{0\fsf},\dt_{\fsf\obs}^{\pre}).
\end{align}

By \autoref{def_observer_RTautomata}, an observer of $\Acalnf$ is denoted by 
\begin{align}\label{def_observer_RTautomata_nf}
  \Acal_{\nf\obs}=(X_{\nf},\Sig_{\obs},x_{0\nf},\dt_{\nf\obs}),
\end{align}
an observer of $\Acalf$ is denoted by
\begin{align}\label{def_observer_RTautomata_f}
  \Acal_{\fsf\obs}=(X_{\fsf},\Sig_{\obs},x_{0\fsf},\dt_{\fsf\obs}).
\end{align}

We extend the notion of recovery point from LFSAs given in 
\cite{Fabre2022ResilienceDES} to LRTAs given as follows.
\begin{definition}\label{def_RP_LRTA}
  A state pair $(q,q')$ is called a \emph{recovery point of type E} (equality) if $L_q(\Acal)=L_{q'}(\Acal)$,
  called a \emph{recovery point of type I} (inclusion) if $L_q(\Acal) \subset L_{q'}(\Acal)$.
  The set of recovery points of type E (resp., I) is denoted as $\Rcal_E$ (resp., $\Rcal_I$).
\end{definition}

When an LRTA $\Acal$ enters $q$ or $q'$ at the same time, where $q$ is reached via a faulty run, but $q'$ is reached via a non-faulty run,
from then on, one cannot be sure whether the future behavior is generated from $q$ because $L_q(\Scal)$ is equal to or contained in $L_{q'}(\Acal)$, in other words, one cannot be sure whether the future behavior is generated from a faulty run, resulting in that the LRTA $\Acal$ returns to normalcy even if some faulty event has occurred. 

Because for a (unlabeled) real-time automaton, its language complementation is computable \cite{Dima2002RealTimeAutomata},
and for two general time automata, their language intersection problem is decidable \cite{Alur1994TimedAutomaton},
we have the set $\Rcal_E$ and the set $\Rcal_I$ are both computable for an LRTA $\Acal$.
From now on we only mention $\Rcal$, which means $\Rcal_E$ or $\Rcal_I$ but not both.

\begin{definition}\label{def_Resilience_LRTA}
  Consider an LRTA $\Acal$, a set $\Ef$ of faulty events, a set $\Rcal$ of recovery points, and a nonnegative integer $n$.
  A faulty run $q_0^1\xrightarrow[]{s_1}q_1=:\pi_1$, where $q_0^1\in Q_0$, $s_1\in (E\times \R_{\ge0})^* (E_o\times \R_{\ge0}) (E_{uo}\times \{0\})^*$, is called \emph{$n$-resilient with respect to $\Rcal$} if 
  for every run $q_1\xrightarrow[]{\bar s_2}q_2$, where $\bar s_2\in (E\times \R_{\ge0})^* (E_o\times \R_{\ge0}) (E_{uo}\times \{0\})^*$, $|\bar s_2|\ge n$ implies 
  there is a prefix $q_1\xrightarrow[]{s_2}q^1=:\pi_2$ of $q_1\xrightarrow[]{\bar s_2}q_2$ and
  a run $q_0^2\xrightarrow[]{r} q^2=:\pi_3$, where $\pi_3$ is in $\Acalnf$,
  $q_0^2\in Q_0$,
  $s_2\in (E_{uo}\times\{0\})^*\cup (E\times \R_{\ge0})^* (E_o\times \R_{\ge0}) (E_{uo}\times \{0\})^*$,
  $r\in (\Enf\times \R_{\ge0})^* (\Eonf\times \R_{\ge0}) (\Euonf\times \{0\})^*$,
  such that
  $\WEIGHT(\pi_1\pi_2)=\WEIGHT(\pi_3)$, $\ell(\tau(\pi_1\pi_2)) = \ell(\tau(\pi_3))$,
  and $(q^1,q^2)\in\Rcal$. The LRTA $\Acal$ is called \emph{resilient with respect to 
  $\Ef$ and $\Rcal$} if for every faulty run $q_0^1\xrightarrow[]{s_1}q_1$, where $q_0^1\in Q_0$,
  $s_1\in (E\times \R_{\ge0})^* (E_o\times \R_{\ge0}) (E_{uo} \times \{0\})^*$,
  there exists a nonnegative integer $m$ such that $q_0^1\xrightarrow[]{s_1}q_1$ is
  $m$-resilient with respect to $\Rcal$.
\end{definition}

By \autoref{def_Resilience_LRTA}, if a faulty run is resilient, then after sufficiently 
many events occur, the future behavior will be equal to (or contained in) another future behavior
via a nonfaulty run such that the two runs generate the same timed labeled sequence.
If an LRTA $\Acal$ is resilient, then whenever a faulty event occurs, after sufficiently
many events occur, the future behavior will be equal to (or contained in) another future behavior
via a nonfaulty run such that the two runs generate the same timed labeled sequence
(observationally equivalent as in \autoref{fig:resilienceDES});
roughly speaking, whenever a faulty event occurs, after sufficiently many events occur, the behavior returns to normalcy, and the occurrence of the faulty event cannot be observed.

In \autoref{def_Resilience_LRTA}, all runs $\pi_1,\pi_2,\pi_3$ satisfy that after the 
last observable event, the runs are unobservable and instantaneous.
This condition implies that possible recovery points only appear at the instants of observations. 
Such a setting makes the explanation of resilience more intuitive and makes
the subsequent discussion and verification unnecessarily complicated. 

			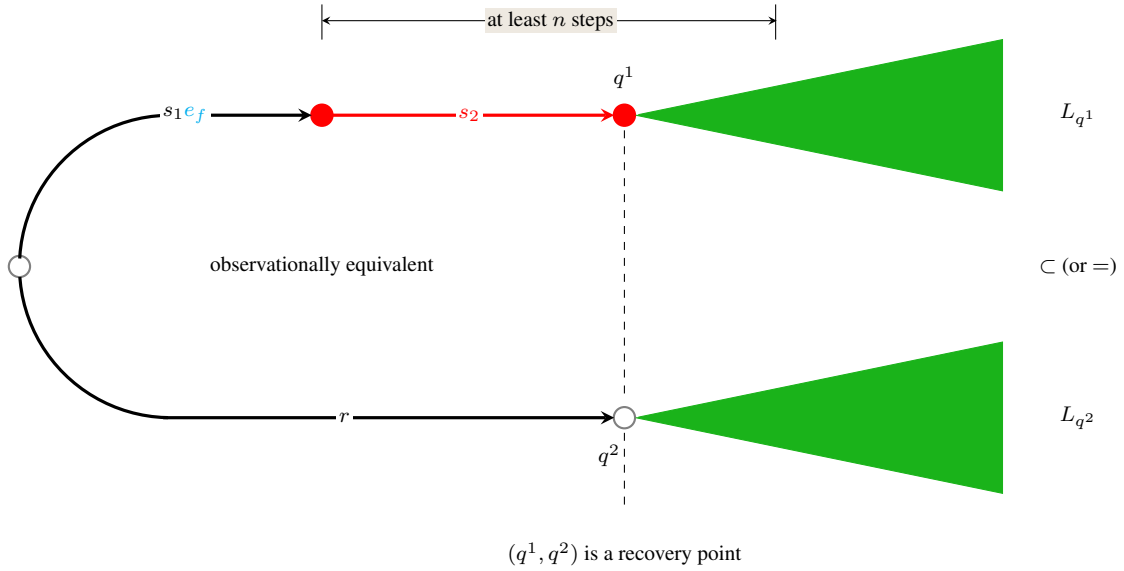
\begin{figure}[!htbp]
			  \centering
			  \begin{tikzpicture}[>=stealth',shorten >=1pt,node distance=2.0 pt, scale = 1.0, transform shape,>=stealth,inner sep=2pt, every text node part/.style={align=center}]
				\scriptsize
				\draw[gray, thick] (0,0) circle (4pt);
				\filldraw[red, thick] (4,2) circle (4pt);
				\filldraw[red, thick] (8,2) circle (4pt);
				\draw[gray, thick] (8,-2) circle (4pt);

				\draw [very thick, domain=90:178] plot ({2+2*cos(\x)},{2*sin(\x)});
				
				\draw [->, very thick, red] (4,2) -- node [inner sep=1pt, minimum size=1pt, midway, fill=white] {$s_2$} (7.9,2);

				\node at (8,2.5) {$q^1$};

				\draw [->, very thick] (2,2) -- node [inner sep=1pt, minimum size=1pt, midway, fill=white, pos=0.1] {$s_1\ef$} (3.9,2);

				\draw [very thick, domain=184:270] plot ({2+2*cos(\x)},{2*sin(\x)});

				\draw [->, very thick]
				(1.9,-2) -- node [inner sep=1pt, minimum size=1pt, midway, fill=white, pos=0.4] {$r$} (7.9,-2)
				;

				\node at (7.8,-2.5) {$q^2$};

	\draw [-, dashed] (8,1.8) -- (8,-1.8);
	\draw [-, dashed] (8,-2.2) -- (8,-3.2);

	\node at (8,-3.8) {$(q^1,q^2)$ is a recovery point};
	\node at (14,2) {$L_{q^1}$};
	\node at (14,0) {$\subset$ (or $=$)};
	\node at (14,-2) {$L_{q^2}$};

	\draw [-,fill=green,green]
	(8.15,2) -- (13,1) -- (13,3) -- (8.15,2);

	\draw [-,fill=green,green]
	(8.15,-2) -- (13,-3) -- (13,-1) -- (8.15,-2);

	\draw [-] (4,3) -- (4,3.5);
	\draw [-] (10,3) -- (10,3.5);
	\draw [<->] (4,3.25) -- node [inner sep=1pt, minimum size=1pt, fill=Pantone7527] {at least $n$ steps} (10.05,3.25);

	\node at (4,0) {observationally equivalent};

	\end{tikzpicture}
	\caption{An intuitive description for resilience. Once a faulty event (e.g., $\ef$) occurs, after at least $n$ more
	  events occur, the system returns to normalcy (in state $q^1$), where the normalcy means that there exists a state
	  $q^2$ such that the system behavior starting from $q^1$ (that is, the generated language $L_{q^1}$) is contained in
	  (or equal to) the system behavior starting from $q^2$ (that is, language $L_{q^2}$), $q^2$ can be reached from a run
	with no faulty events, and the runs leading to $q^1$ and $q^2$ cannot be distinguished by observation.}
  \label{fig:resilienceDES}
\end{figure}

\subsection{A necessary and sufficient condition for resilience}

In order to obtain a concise necessary and sufficient condition for the definition of resilience, we characterize its negation.

\begin{lemma}\label{lem1_Resilience_LRTA}
  Consider an LRTA $\Acal$, a set $\Ef$ of faulty events, and a set $\Rcal$ of recovery points.
  The LRTA $\Acal$ is not resilient with respect to $\Ef$ and $\Rcal$, if and only if,
  there exists a faulty run $q_0^1\xrightarrow[]{s_1}q_1$, where $q_0^1\in Q_0$, $s_1\in (E\times \R_{\ge0})^* (E_o\times \R_{\ge0}) (E_{uo}\times \{0\})^*$ such that 
  $q_0^1\xrightarrow[]{s_1}q_1$
  is not $n$-resilient for any nonnegative integer $n$, if and only if, there exists a faulty
  run $q_0^1\xrightarrow[]{s_1}q_1=:\pi_1$, where $q_0^1\in Q_0$, $s_1\in (E\times \R_{\ge0})^* (E_o\times \R_{\ge0}) (E_{uo}\times \{0\})^*$,
  for every nonnegative
  integer $n$, there exists a run
  $q_1\xrightarrow[]{\bar s_2}q_2$, where $\bar s_2\in (E\times \R_{\ge0})^* (E_o\times \R_{\ge0}) (E_{uo}\times \{0\})^*$,
  $|\bar s_2|\ge n$, such that
  for every prefix $q_1\xrightarrow[]{s_2}q^1=:\pi_2$ of $q_1\xrightarrow[]{\bar s_2}q_2$ and
  every run $q_0^2\xrightarrow[]{r} q^2=:\pi_3$, where $\pi_3$ is in $\Acalnf$,
  $q_0^2\in Q_0$,
  $s_2\in (E_{uo}\times\{0\})^*\cup (E\times \R_{\ge0})^* (E_o\times \R_{\ge0}) (E_{uo}\times \{0\})^*$,
  $r\in (\Enf\times \R_{\ge0})^* (\Eonf\times \R_{\ge0}) (\Euonf\times \{0\})^*$,
  $\WEIGHT(\pi_1\pi_2)=\WEIGHT(\pi_3)$ and $\ell(\tau(\pi_1\pi_2)) = \ell(\tau(\pi_3))$
  imply $(q^1,q^2)\not\in\Rcal$. 
\end{lemma}

By \autoref{lem1_Resilience_LRTA}, we get a further result for the negation of resilience.

\begin{theorem}\label{thm1_Resilience_LRTA}
  Consider an LRTA $\Acal$, a set $\Ef$ of faulty events, and a set $\Rcal$ of recovery points.
  The LRTA $\Acal$ is not resilient with respect to $\Ef$ and $\Rcal$, if and only if,
  there exists a run 
  \begin{align}\label{eqn2_Resilience_LRTA}
	q_0 \xrightarrow[]{s_1} q_1 \xrightarrow[]{s_2} \cdots \xrightarrow[]{s_{n}} q_n \xrightarrow[]{s_{n+1}} \cdots \xrightarrow[]{s_{n+m}} q_{n+m} \xrightarrow[]{s_{n+m+1}} 
	\cdots \xrightarrow[]{s_{n+m+p}} q_{n+m+p}, 
  \end{align}
  where $q_0\in Q_0$, $s_i\in (E_{uo}\times \R_{\ge0})^* (E_o\times \R_{\ge0}) (E_{uo}\times \{0\})^*$, $i\in \llb 1, n \rrb $,
  $q_0 \xrightarrow[]{s_1} q_1 \xrightarrow[]{s_2} \cdots \xrightarrow[]{s_{n}} q_n$ is faulty,
  $s_j\in (E_{uo}\times \R_{\ge0})^* (E_o\times \R_{\ge0}) $, $j\in \llb n+1, n+m+p \rrb$,
  $q_{n+m}=q_{n+m+p}$,
  $|s_{n+m+1}\dots s_{n+m+p}|>0$\footnote{This implies $p>0$.}, such that in the pre-observer $\Acal_{\nf\obs}^{\pre}$ of $\Acalnf$, the run
  \begin{align}\label{eqn3_Resilience_LRTA}
	x_{0\nf} \xrightarrow[]{\ell(\tau(q_0 \xrightarrow[]{s_1} q_1))} \cdots \xrightarrow[]{\ell(\tau( q_{n+m-1} \xrightarrow[]{s_{n+m}} q_{n+m} ))} x_{n+m} \cdots \xrightarrow[]{\ell(\tau( q_{n+m+p-1} \xrightarrow[]{s_{n+m+p}} q_{n+m+p} ))} x_{n+m+p} 
  \end{align}
  satisfies that $x_{n+m}=x_{n+m+p}$, for every $i\in \llb 0,m+p-1 \rrb$ and every prefix 
	\[
	  q_{n+i} \xrightarrow[]{s_{n+i+1}'} q_{n+i}'
	\]
	of the run $q_{n+i} \xrightarrow[]{s_{n+i+1}} q_{n+i+1}$ in \eqref{eqn2_Resilience_LRTA} 
  with $s_{n+i+1}' \in (E_{uo}\times\{0\})^*$, $(\{q_{n+i}'\} \times x_{n+i}) \cap \Rcal = \emptyset$.
\end{theorem}

\begin{proof}
  ``only if'': Assume $\Acal$ is not resilient with respect to $\Ef$ and $\Rcal$. By \autoref{lem1_Resilience_LRTA},
  there exists a faulty run
  \begin{align}\label{eqn6_Resilience_LRTA} 
	q_0 \xrightarrow[]{s_1} q_1 \xrightarrow[]{s_2} \cdots \xrightarrow[]{s_{n}} q_n =: \pi_1,
  \end{align}
  where $q_0 \in Q_0$, $s_i\in (E_{uo}\times \R_{\ge0})^* (E_o\times \R_{\ge0}) (E_{uo}\times \{0\})^*$, $1\le i\le n$; choose a sufficiently large positive integer $t$, there exists a run
  \begin{align}\label{eqn7_Resilience_LRTA}
	q_n \xrightarrow[]{s_{n+1}} \cdots \xrightarrow[]{s_{n+t}} q_{n+t} =: \pi_2,
  \end{align}
  where $s_{n+j}\in (E_{uo}\times \R_{\ge0})^* (E_o\times \R_{\ge0})$, $1\le j\le t$, then $|s_{n+1} \dots s_{n+t}| \ge t$; for every $k\in \llb 0,t-1 \rrb$, for every prefix 
	\begin{align}
	  q_{n+k} \xrightarrow[]{s_{n+k+1}'} q_{n+k}' =: \pi_3(k)
	\end{align}
	of $q_{n+k} \xrightarrow[]{s_{n+k+1}} q_{n+k+1}$ 
	with $s_{n+k+1}' \in (E_{uo}\times\{0\})^*$,
  for every run 
  \begin{align}
	q_0^2\xrightarrow[]{r} q^2 =: \pi_4,
  \end{align}
  in $\Acalnf$, where $q_0^2\in Q_0$,
  $r\in (\Enf\times \R_{\ge0})^* (\Eonf\times \R_{\ge0}) (\Euonf\times \{0\})^*$,
  we have 
  $\WEIGHT(\pi_1\pi_3(k))=\WEIGHT(\pi_4)$ and $\ell(\tau(\pi_1\pi_3(k))) = \ell(\tau(\pi_4))$
  imply $(q_{n+k}',q^2)\not\in\Rcal$. This implies that 
  in the pre-observer $\Acal_{\nf\obs}^{\pre}$ of $\Acalnf$, in the run
  \begin{align*}
	x_{0\nf} \xrightarrow[]{\ell(\nu(\pi_1))} x_n \xrightarrow[]{\ell(\nu(\pi_3(k)))} x_{n+k},
  \end{align*}
  $(\{q_{n+k}'\} \times x_{n+k}) \cap \Rcal = \emptyset$, $k\in \llb 0,t-1 \rrb$,
  where function $\ell\circ\nu$ is defined in \eqref{eqn1_Resilience_LRTA}. Because $t$ is
  sufficiently large, by the pigeonhole principle, there exists $0 \le m < m+p \le t$, a prefix
  \begin{align*}
	q_0 \xrightarrow[]{s_1\dots s_n} q_n  \xrightarrow[]{s_{n+1}\dots s_{n+m}} q_{n+m}  \xrightarrow[]{s_{n+m+1}\dots s_{n+m+p}} q_{n+m+p}
  \end{align*}
  of $\pi_1\pi_2$ such that $q_{n+m} = q_{n+m+p}$ and $x_{n+m} = x_{n+m+p}$, where $x_{n+m}$ and $x_{n+m+p}$ are states of $\Acal_{\nf\obs}^{\pre}$ in the run 
  \begin{align*}
	x_{0\nf} \xrightarrow[]{\ell(\nu( q_0 \xrightarrow[]{s_1\dots s_n} q_n  \xrightarrow[]{s_{n+1}\dots s_{n+m}} q_{n+m} ))} x_{n+m} \xrightarrow[]{\ell(\nu(  q_{n+m}  \xrightarrow[]{s_{n+m+1}\dots s_{n+m+p}} q_{n+m+p} ))} x_{n+m+p}.
  \end{align*}

  ``if'': Assume that the condition after ``if and only if'' of \autoref{thm1_Resilience_LRTA} holds. Choose a faulty run 
  \begin{align}\label{eqn9_Resilience_LRTA}
	q_0 \xrightarrow[]{s_1} q_1 \xrightarrow[]{s_2} \cdots \xrightarrow[]{s_{n}} q_n =: \pi_1
  \end{align}
  as in \eqref{eqn2_Resilience_LRTA}. Choose an arbitrary positive integer $t$ and construct a run
  \begin{align}\label{eqn10_Resilience_LRTA}
	q_n \xrightarrow[]{s_{n+1}} \cdots \xrightarrow[]{s_{n+m}} q_{n+m} \left( \xrightarrow[]{s_{n+m+1}} 
	\cdots \xrightarrow[]{s_{n+m+p}} q_{n+m+p}  \right) ^t 
  \end{align}
  from \eqref{eqn2_Resilience_LRTA}.
  Then $|s_{n+1} \dots s_{n+m}\left( s_{n+m+1} \dots s_{n+m+p} \right)^t| \ge t$. Rewrite 
  \eqref{eqn10_Resilience_LRTA} as
  \begin{subequations}\label{eqn11_Resilience_LRTA}
  \begin{align}
	&q_n \xrightarrow[]{s_{n+1}} \cdots \xrightarrow[]{s_{n+m}} q_{n+m} 
	\xrightarrow[]{s_{n+m+1}} \cdots \xrightarrow[]{s_{n+m+p}} q_{n+m+p}\\
	&\xrightarrow[]{s_{n+m+p+1}} \cdots \xrightarrow[]{s_{n+m+2p}} q_{n+m+2p}
	\cdots
	\xrightarrow[]{s_{n+m+(t-1)p+1}} \cdots \xrightarrow[]{s_{n+m+tp}} q_{n+m+tp},
	\end{align}
    \end{subequations}			
  where $s_{n+m+i}=s_{n+m+p+i}=\cdots= s_{n+m+(t-1)p+i}$,
  $q_{n+m+i}=q_{n+m+p+i}=\cdots= q_{n+m+(t-1)p+i}$, $i \in \llb 1,p \rrb$.

  Choose a prefix of \eqref{eqn11_Resilience_LRTA} as
  \[
	q_n \xrightarrow[]{s_{n+1}} \cdots \xrightarrow[]{s_{n+j-1}}q_{n+j-1}\xrightarrow[]{s'_{n+j}} q_{n+j-1}' =: \pi_2(j),
  \]
  where $1\le j \le m+tp$, $s_{n+j}'\in (E_{uo}\times \{0\})^*$,
  then for every run
  \begin{align}
	q_0^2\xrightarrow[]{r} q^2=:\pi_3,
  \end{align}
  in $\Acalnf$ with $q_0^2\in Q_0$ and
  $r\in (\Enf\times \R_{\ge0})^* (\Eonf\times \R_{\ge0}) (\Euonf\times \{0\})^*$,
  we have
  $\WEIGHT(\pi_1\pi_2(j))=\WEIGHT(\pi_3)$ and $\ell(\tau(\pi_1\pi_2(j))) = \ell(\tau(\pi_3))$
  imply $(q_{n+j}',q^2)\not\in\Rcal$ by definition of pre-observer.
  By \autoref{lem1_Resilience_LRTA}, $\Acal$ is not resilient with respect to $\Ef$ and $\Rcal$.
\end{proof}

\subsection{A verification algorithm for resilience}

\autoref{thm1_Resilience_LRTA} provides a necessary and sufficient condition for the negation of resilience, 
but the condition is not verifiable, because the pre-observer $\Acal_{\nf\obs}^{\pre}$ of 
$\Acal_{\nf}$ is not computable. Next, we use an observer $\Acal_{\nf\obs}$ of $\Acal_{\nf}$
and the concurrent composition $\CC(\Acal,\Acal_{\nf\obs})$ to give a verification algorithm
for the negation of resilience, where $\Acal_{\nf\obs}$ and $\CC(\Acal,\Acal_{\nf\obs})$
are computable and can be used to verify the negation of resilience.

The condition in \autoref{thm1_Resilience_LRTA} will be verifiable, if one can compute for every $i$, the pair of the run $q_i\xrightarrow[]{s_{i+1}} q_{i+1}$ (as in \eqref{eqn2_Resilience_LRTA}) in $\Acal$ and the run $x_i \xrightarrow[]{\ell(\tau(q_i\xrightarrow[]{s_{i+1}} q_{i+1}))} x_{i+1}$ (as in \eqref{eqn3_Resilience_LRTA}) in $\Acal_{\nf\obs}^{\pre}$.
However, it is not possible to compute all such pairs, because there are infinitely many such pairs.
Despite of this impossibility, one can check for all states $q_i,q_{i+1}$ of $\Acal$ and all state sets $x_i,x_{i+1}$ of $\Acal_{\nf\obs}^{\pre}$, if there are such runs from $q_i$ to $q_{i+1}$ and from $x_i$ to $x_{i+1}$ as above instead.
Formally, one needs to check the following. 
Choose two arbitrary states $q_1,q_2\in Q$ and two arbitrary state subsets $x_1,x_2\subset Q$.
Check whether \eqref{quoteA_Resilience_LRTA} is satisfied.
\begin{equation*}
  \tag{A}\label{quoteA_Resilience_LRTA}
  \parbox{\dimexpr\linewidth-4em}{%
	\strut
	There are runs
	\begin{align*}
	  & q_1 \xrightarrow[]{s_2} q_2\text{ in }\Acal,\\
	  & x_1 \xrightarrow[]{\ell(\tau(q_1\xrightarrow[]{s_2} q_2))} x_2\text{ in }\Acal_{\nf\obs}^{\pre},
	\end{align*}
	where $s_2\in (E_{uo}\times \R_{\ge0})^* (E_o\times \R_{\ge0}) (E_{uo}\times \{0\})^*$.
	\strut
  }
\end{equation*}

Equivalently, one needs to check for two arbitrary states $q_1,q_2\in Q$ and two arbitrary state subsets $x_1,x_2\subset Q$, whether \eqref{quoteB_Resilience_LRTA} is satisfied.
\begin{equation*}
  \tag{B}\label{quoteB_Resilience_LRTA}
  \parbox{\dimexpr\linewidth-4em}{%
	\strut
	There exists $\sigma\in\Sig$ and $t\in\Q_{\ge0}$ such that there is a run
	\begin{align*}
	  & q_1 \xrightarrow[]{s_2} q_2\text{ in }\Acal,
	\end{align*}
	where $s_2\in (E_{uo}\times \R_{\ge0})^* (E_o\times \R_{\ge0}) (E_{uo}\times \{0\})^*$, $\ell(\tau(q_1 \xrightarrow[]{s_2} q_2))=(\sigma,t)$;
	for every $q_1'\in x_1$ and every $q_2'\in x_2$, there exists a run
	\begin{align*}
	  & q_1' \xrightarrow[]{s_2'} q_2'\text{ in }\Acalnf,
	\end{align*}
	where $s_2'\in (E_{uo\nf}\times \R_{\ge0})^* (E_{o\nf}\times \R_{\ge0}) (E_{uo\nf}\times \{0\})^*$, $\ell(\tau(q_1' \xrightarrow[]{s_2'} q_2'))=(\sigma,t)$;
	for every $q_1'\in x_1$ and every $q_2'\in Q\setminus x_2$, there exists no run
	\begin{align*}
	  & q_1' \xrightarrow[]{s_2'} q_2'\text{ in }\Acalnf,
	\end{align*}
	such that $s_2'\in (E_{uo\nf}\times \R_{\ge0})^* (E_{o\nf}\times \R_{\ge0}) (E_{uo\nf}\times \{0\})^*$, $\ell(\tau(q_1' \xrightarrow[]{s_2'} q_2'))=(\sigma,t)$.
	\strut
  }
\end{equation*}

  \begin{lemma}\label{lem2_Resilience_LRTA}
  Choose two arbitrary states $q_1,q_2\in Q$ and two arbitrary state subsets $x_1,x_2\subset Q$. \eqref{quoteA_Resilience_LRTA} is satisfied if and only if \eqref{quoteB_Resilience_LRTA} is satisfied.
\end{lemma}
\begin{proof}
  By definition, if \eqref{quoteA_Resilience_LRTA} is satisfied, then \eqref{quoteB_Resilience_LRTA} is satisfied except that $t\in \R_{\ge0}$. Because $\Q_{\ge0}$ is dense in $\R_{\ge0}$, $t$ could be chosen in $\Q_{\ge0}$. See \autoref{fig4_Resilience_LRTA} as an illustration.
  It is easy to see that if \eqref{quoteB_Resilience_LRTA} is satisfied, then \eqref{quoteA_Resilience_LRTA} is also satisfied.

  	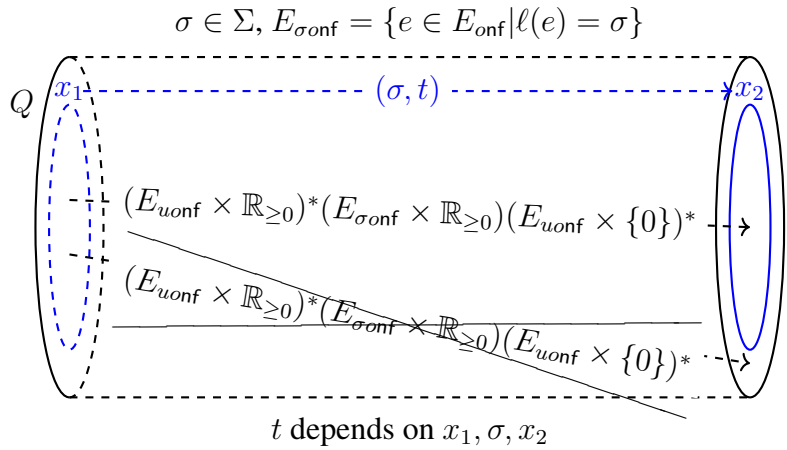
\begin{figure}[!htbp]
	  \centering
			\begin{tikzpicture}
			  [thick, scale = 1.8, align = center]
		
                \draw (5,0) ellipse (0.25 and 1.25);
                \draw (0,1.25) arc (90:270:0.25 and 1.25);
                \draw [dashed] (0,1.25) arc (90:-90:0.25 and 1.25);
                \node at (-.35,0.9) {$Q$};
                \draw [dashed] (0,1.25) -- (5,1.25);
                \draw [dashed] (0,-1.25) -- (5,-1.25);
                
                \draw [dashed, blue] (0,0) ellipse (0.15 and 0.90);
                \node at (0.0,1.0) {$\blue x_1$};
				\node at (2.5, -1.5) {$t$ depends on $x_1,\sigma,x_2$};

                \draw [blue] (5,0) ellipse (0.15 and 0.90);
                \node at (5.0,1.0) {$\blue x_2$};
				\draw [dashed, ->, blue] (0.1,1.0) -- (4.87,1.0) node [midway, fill=white] {$(\sigma,t)$};

				\node at (2.5, 1.5) {$\sigma\in\Sigma$, $E_{\sigma o\nf}=\{e\in E_{o\nf}|\ell(e)=\sigma\}$};
                \draw [dashed, ->] (0,0.2) -- (5,0.0) node [sloped, midway, fill=white] {$(E_{uo\nf}\times \R_{\ge0})^* (E_{\sigma o\nf}\times \R_{\ge0}) (E_{uo\nf}\times \{0\})^*$};

				\draw [dashed, ->] (0.0,-0.2) -- (5.0,-1.0) node [sloped, midway, fill=white] {$\xcancel{(E_{uo\nf}\times \R_{\ge0})^* (E_{\sigma o\nf}\times \R_{\ge0}) (E_{uo\nf}\times \{0\})^*}$};

                \end{tikzpicture}
				\caption{An illustration of the proof of \autoref{lem2_Resilience_LRTA}.}
		\label{fig4_Resilience_LRTA} 
	\end{figure}

\end{proof}

One also needs to check the following slight variant of \eqref{quoteB_Resilience_LRTA}.
Consider two arbitrary states $q_1,q_2\in Q$ and two arbitrary state subsets $x_1,x_2\subset Q$.
\begin{equation*}
  \tag{C}\label{quoteC_Resilience_LRTA}
  \parbox{\dimexpr\linewidth-4em}{%
	\strut
	There exists $\sigma\in\Sig$ and $t\in\Q_{\ge0}$ such that there is a run
	\begin{align*}
	  & q_1 \xrightarrow[]{s_2} q_2 \text{ in }\Acal,
	\end{align*}
	where $s_2\in (E_{uo}\times \R_{\ge0})^* (E_o\times \R_{\ge0})$, $\ell(\tau(q_1 \xrightarrow[]{s_2} q_2))=(\sigma,t)$;
	for every $q_1'\in x_1$ and every $q_2'\in x_2$, there exists a run
	\begin{align*}
	  & q_1' \xrightarrow[]{s_2'} q_2'\text{ in }\Acalnf,
	\end{align*}
	where $s_2'\in (E_{uo\nf}\times \R_{\ge0})^* (E_{o\nf}\times \R_{\ge0}) (E_{uo\nf}\times \{0\})^*$, $\ell(\tau(q_1' \xrightarrow[]{s_2'} q_2'))=(\sigma,t)$;
	for every $q_1'\in x_1$ and every $q_2'\in Q\setminus x_2$, there exists no run
	\begin{align*}
	  & q_1' \xrightarrow[]{s_2'} q_2'\text{ in }\Acalnf,
	\end{align*}
	such that $s_2'\in (E_{uo\nf}\times \R_{\ge0})^* (E_{o\nf}\times \R_{\ge0}) (E_{uo\nf}\times \{0\})^*$, $\ell(\tau(q_1' \xrightarrow[]{s_2'} q_2'))=(\sigma,t)$; for every state $q$ in the longest prefix
	$q_2 \xrightarrow[]{\bar s_2} \bar q_2$ of $q_1 \xrightarrow[]{s_2} q_2$ with $\bar s_2 \in (E_{uo}\times \{0\})^*$,
	$(\{q\}\times x_1) \cap \Rcal = \emptyset$.
	\strut
  }
\end{equation*}

According to a slight variant of the procedure given in \cite{Zhang2024OpacityReal-TimeAutomata} for computing observers of $\Acal$, the following result holds.
\begin{proposition}\label{prop1_Resilience_LRTA}
  The satisfiability of \eqref{quoteB_Resilience_LRTA} and \eqref{quoteC_Resilience_LRTA} can be verified in $\mathsf{N}$2$\EXPTIME$ in the size of $\Acal$.
\end{proposition}

Now we are ready to give a tool to verify the negation of resilience.

\begin{definition}\label{def_CC-RTA-RTAObs} 
  Consider an LRTA $\Acal$ and a set $\Ef$ of faulty events. Define a concurrent composition\footnote{This concurrent composition can be regarded as a natural and nontrivial
  extension of the concurrent composition of an LFSA and its observer as in 
\cite[eqn.~(16)]{Zhang2023UnifiedFrame4DES}.}
  \begin{align}\label{CC-RTA-RTAObs}
	\CC(\Acal,\Acal_{\nf\obs}) = (Q\times 2^Q, \dt,Q_0\times\{x_{0\nf}\}),
  \end{align}
  where $Q\times 2^Q$ is the set of \emph{states}, $\dt\subset (Q\times 2^Q)\times (Q\times 2^Q)$ is the \emph{transition relation}, $Q_0\times\{x_{0\nf}\}$ is the set of \emph{initial states}.
  
  For every two states $(q_1,x_1),(q_2,x_2)\in Q\times 2^Q$, $((q_1,x_1),(q_2,x_2)) \in \dt$ if and only if $q_1,q_2,x_1,x_2$ satisfy \eqref{quoteB_Resilience_LRTA}. Define $\dt_{\fsf}\subset \dt$ as follows: for every two states $(q_1,x_1),(q_2,x_2)\in Q\times 2^Q$, $((q_1,x_1),(q_2,x_2)) \in \dt_{\fsf}$ if and only if $q_1,q_2,x_1,x_2$ satisfy \eqref{quoteB_Resilience_LRTA} and there is a faulty event in $s_2$ as in \eqref{quoteB_Resilience_LRTA}. Transitions of $\dt_{\fsf}$ are called \emph{faulty}. A run is called \emph{faulty} if it contains at least one faulty transition.
  Define $\dt_{\eqref{quoteC_Resilience_LRTA}}$ as the set of 
  $((q_1,x_1),(q_2,x_2))$ such that $q_1,q_2,x_1,x_2$ satisfy \eqref{quoteC_Resilience_LRTA}.
\end{definition}

\begin{remark}\label{rem1_Resilience_LRTA}
  Note that an LRTA $\Acal$ may have infinitely many concurrent compositions $\CC(\Acal,\Acal_{\nf\obs})$ because
  $\Acal_{\nf}$ may have infinitely many observers $\Acal_{\nf\obs}$.
\end{remark}

By \autoref{prop1_Resilience_LRTA}, the following result holds.

\begin{theorem}\label{thm2_Resilience_LRTA}
  A concurrent composition $\CC(\Acal,\Acal_{\nf\obs})$ can be computed in $\mathsf{N}$2$\EXPTIME$ in the size of $\Acal$. $\dt_{\fsf}$, and $\dt_{\eqref{quoteC_Resilience_LRTA}}$ can also be computed in $\mathsf{N}$2$\EXPTIME$.
\end{theorem}

By $\CC(\Acal,\Acal_{\nf\obs})$ and \autoref{thm1_Resilience_LRTA}, we give a verifiable necessary and sufficient condition for the negation of resilience of $\Acal$.

\begin{theorem}\label{thm3_Resilience_LRTA}
  Consider an LRTA $\Acal$, a set $\Ef$ of faulty events, and a set $\Rcal$ of recovery points.
  The LRTA $\Acal$ is not resilient with respect to $\Ef$ and $\Rcal$, if and only if, in a concurrent composition $\CC(\Acal,\Acal_{\nf\obs})$ there is a run
  \begin{align}\label{eqn4_Resilience_LRTA}
	(q_0,x_{0\nf}) \to^* (q_1,x_1) \to^* (q_2,x_2) \to^* (q_2,x_2),
  \end{align}
where $(q_0,x_{0\nf}) \to^* (q_1,x_1)$ is a faulty run, the transitions of the run $(q_1,x_1) \to^* (q_2,x_2) \to^* (q_2,x_2)$ belongs to $\dt_{\eqref{quoteC_Resilience_LRTA}}$.
\end{theorem}

\begin{proof}
  By comparing the details of $\CC(\Acal,\Acal_{\nf\obs})$ and the condition in \autoref{thm1_Resilience_LRTA}, this result holds. 
\end{proof}

\begin{remark}\label{rem2_Resilience_LRTA}
  Although an LRTA $\Acal$ may have infinitely many concurrent compositions $\CC(\Acal,\Acal_{\nf\obs})$, any concurrent composition can verify the resilience of $\Acal$ by the proof of \autoref{thm3_Resilience_LRTA}.
\end{remark}

By \autoref{thm2_Resilience_LRTA} and \autoref{thm3_Resilience_LRTA}, the negation of resilience of an LRTA $\Acal$ can 
be verified in $\mathsf{N}$2$\EXPTIME$.

\begin{example}\label{exam1_Resilience_LRTA}
\refstepcounter{couRTA}
	Consider the LRTA $\Acal_{\thecouRTA\label{couRTA4}}$ shown in \autoref{fig1_Resilience_LRTA}, where $\Ef=\{\ef\}$, 
	$E_{uo}=\{u,\ef\}$, $E_{o}=\{a\}$, $\ell(a)=a$. The subautomaton $\Acal_{\ref{couRTA4}\nf}$ is shown in \autoref{fig2_Resilience_LRTA}. Consider $\Rcal_I=\{(q_4,q_5)\}$.

	We next check whether $\Acal_{\ref{couRTA4}}$ is resilient with respect to 
	$\Ef$ and $\Rcal_I$. Consider faulty run $q_0 \xrightarrow[]{\ef/1} q_1
	\xrightarrow[]{a/1} q_4$, its timed label sequence is $(a,2)$. In addition, one has
	$\Mt(\Acal_{\ref{couRTA4}\nf},(a,2))=\{q_6\}$. One has $(q_4,q_6)\notin\Rcal_I$.
	In addition, one can find the following two runs in $\Acal$ and the pre-observer $\Acal_{\nf\obs}^{\pre}$, respectively:

	\begin{center}
	  \begin{tikzcd}[column sep=18pt, row sep=small]
		\text{in }\Acal: & q_0 \arrow[r, "\ef/1"] & q_1 \arrow[r, "a/1"] & q_4 \arrow[r, "u/0"] & q_7 \arrow[r, "a/1"] & q_7 \arrow[r, "a/1"] & q_7 \arrow[r, "a/1"] & \cdots, \\
		\text{in }\Acal_{\nf\obs}^{\pre}: & \{q_0\} \arrow[rr, "(a{,}2)"]  && \{q_6\} \arrow[rr, "(a{,}1)"]  && \emptyset \arrow[r,"(a{,}1)"] & \emptyset \arrow[r,"(a{,}1)"] & \cdots.
  	  \end{tikzcd}
    \end{center}

  The two runs correspond to \eqref{eqn2_Resilience_LRTA} and \eqref{eqn3_Resilience_LRTA}, respectively, because the $q_4$ here corresponds to the $q_n$ in \eqref{eqn2_Resilience_LRTA}, the $\{q_6\}$ here corresponds to the $x_n$ in \eqref{eqn3_Resilience_LRTA}. For pair $(q_4,\{q_6\})$, one has $(q_4,q_6)\notin\Rcal_I$; for pair $(q_7,\{q_6\})$, one has $(q_7,q_6)\notin\Rcal_I$; for pair $(q_7,\emptyset)$, one has $(\{q_7\}\times\emptyset)\cap \Rcal_I=\emptyset$. Then the condition in \autoref{thm1_Resilience_LRTA} is satisfied, one concludes that $\Acal_{\ref{couRTA4}}$ is not resilient with respect to $\Ef$ and $\Rcal_I$.

	On the other hand, from a part of
	concurrent composition $\CC(\Acal_{\ref{couRTA4}},\Acal_{\ref{couRTA4}\nf\obs})$ as 
	in \autoref{fig3_Resilience_LRTA}, one can find the following run
	\begin{center}
	  \begin{tikzcd}[column sep=18pt]
		(q_0,\{q_0\}) \arrow{r} & (q_4,\{q_6\}) \arrow{r} & (q_7,\emptyset) \arrow{r} & (q_7,\emptyset) \arrow{r} & \cdots
	  \end{tikzcd}
	\end{center}
	which corresponds to \eqref{eqn4_Resilience_LRTA}, because transition $(q_0,\{q_0\}) \to (q_4,\{q_6\})$ is faulty,
	transition $(q_4,\{q_6\}) \to (q_7,\emptyset)$ and transition $(q_7,\emptyset) \to (q_7,\emptyset)$
	belong to $\dt_{\eqref{quoteC_Resilience_LRTA}}$. Then the condition in
	\autoref{thm3_Resilience_LRTA} is satisfied, one also concludes that $\Acal_{\ref{couRTA4}}$ is not resilient with respect to $\Ef$ and $\Rcal_I$.

 	\begin{figure}[!htbp]
	\centering
	\begin{tikzpicture}
	[>=stealth',shorten >=1pt,thick,auto,node distance=2.5 cm, scale = 0.7, transform shape,
	->,>=stealth,inner sep=2pt]

	\tikzstyle{emptynode}=[inner sep=0,outer sep=0]

	\node[initial, state, initial where = left] (q0) {$q_0$};
	\node[state] (q2) [right of = q0] {$q_2$};
	\node[state] (q1) [above of = q2] {$q_1$};
	\node[state] (q3) [below of = q2] {$q_3$};
	\node[state] (q5) [right of = q2] {$q_5$};
	\node[state] (q4) [above of = q5] {$q_4$};
	\node[state] (q6) [below of = q5] {$q_6$};
	\node[state] (q7) [right of = q4] {$q_7$};
	\node[state] (q8) [right of = q5] {$q_8$};

	\path [->]
	(q0) edge node [above, sloped] {$\ef/[1,1]$} (q1)
	(q0) edge node [above, sloped] {$u/[0,1]$} (q2)
	(q0) edge node [above, sloped] {$u/[0,2]$} (q3)
	(q1) edge node [above, sloped] {$a/[0,1]$} (q4)
	(q2) edge node [above, sloped] {$a/[0,0]$} (q5)
	(q3) edge node [above, sloped] {$a/[0,0]$} (q6)
	(q4) edge node [above, sloped] {$u/[0,1]$} (q7)
	(q5) edge node [above, sloped] {$u/[0,1]$} (q8)
	(q7) edge [loop right] node [right] {$a/[0,1]$} (q7)
	(q8) edge [loop right] node [right] {$a/[0,1]$} (q8)
	(q8) edge [loop above] node {$u/[1,1]$} (q8)
	;

	 \end{tikzpicture}
	 \caption{An LRTA $\Acal_{\ref{couRTA4}}$, $\ell(u)=\ell(\ef)=\ep$, $\ell(a)=a$.}
	 \label{fig1_Resilience_LRTA}
   \end{figure}
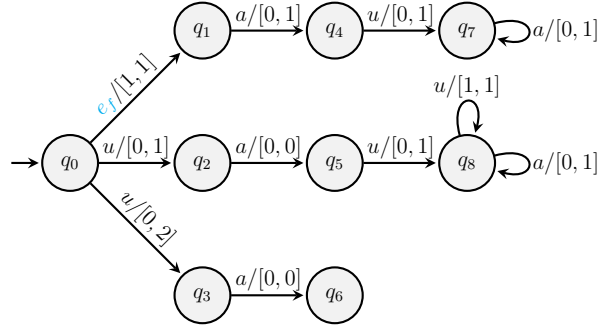

   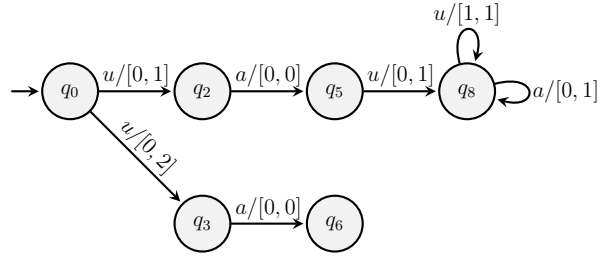
\begin{figure}[!htbp]
	 \centering
	\begin{tikzpicture}
	[>=stealth',shorten >=1pt,thick,auto,node distance=2.5 cm, scale = 0.7, transform shape,
	->,>=stealth,inner sep=2pt]

	\tikzstyle{emptynode}=[inner sep=0,outer sep=0]

	\node[initial, state, initial where = left] (q0) {$q_0$};
	\node[state] (q2) [right of = q0] {$q_2$};
	\node[state] (q3) [below of = q2] {$q_3$};
	\node[state] (q5) [right of = q2] {$q_5$};
	\node[state] (q6) [below of = q5] {$q_6$};
	\node[state] (q8) [right of = q5] {$q_8$};

	\path [->]
	(q0) edge node [above, sloped] {$u/[0,1]$} (q2)
	(q0) edge node [above, sloped] {$u/[0,2]$} (q3)
	(q2) edge node [above, sloped] {$a/[0,0]$} (q5)
	(q3) edge node [above, sloped] {$a/[0,0]$} (q6)
	(q5) edge node [above, sloped] {$u/[0,1]$} (q8)
	(q8) edge [loop right] node [right] {$a/[0,1]$} (q8)
	(q8) edge [loop above] node {$u/[1,1]$} (q8)
	;

	 \end{tikzpicture}
   \caption{The LRTA $\Acal_{\ref{couRTA4}\nf}$, $\ell(u)=\ell(\ef)=\ep$, $\ell(a)=a$.}
	 \label{fig2_Resilience_LRTA}
	\end{figure}

	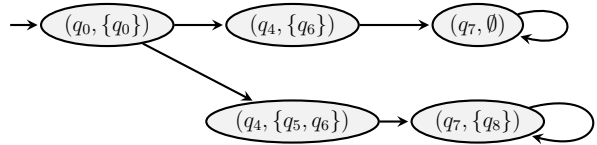
\begin{figure}[!htbp]
		\centering
	\begin{tikzpicture}
	[>=stealth',shorten >=1pt,thick,auto,node distance=3.5 cm, scale = 0.7, transform shape,
	->,>=stealth,inner sep=2pt]

	\tikzstyle{emptynode}=[inner sep=0,outer sep=0]

	\node[initial, elliptic state, initial where = left] (00) {$(q_0,\{q_0\})$};
	\node[elliptic state] (46) [right of = 00] {$(q_4,\{q_6\})$};
	\node[elliptic state] (456) [below = 1cm of 46] {$(q_4,\{q_5,q_6\})$};
	\node[elliptic state] (7phi) [right of = 46] {$(q_7,\emptyset)$};
	\node[elliptic state] (78) [right of = 456] {$(q_7,\{q_8\})$};

	\path [->]
	(00) edge (46)
	(00) edge (456)
	(46) edge (7phi)
	(456) edge (78)
	(7phi) edge [loop right] (7phi)
	(78) edge [loop right] (78)
	;

	 \end{tikzpicture}
	 \caption{A part of a concurrent composition $\CC(\Acal_{\ref{couRTA4}},\Acal_{\ref{couRTA4}\nf\obs})$.}
	 \label{fig3_Resilience_LRTA}
	\end{figure}
\end{example}

  \refstepcounter{couRTA}
\begin{example}\label{exam2_Resilience_LRTA}
  By slightly changing $\Acal_{\ref{couRTA4}}$, we will get a resilient LRTA. After changing the transition
  $q_2 \xrightarrow[]{a/[0,0]} q_5$ to $q_2 \xrightarrow[]{a/[1,1]} q_5$, we get LRTA $\Acal_{\thecouRTA\label{couRTA5}}$
  as in \autoref{fig5_Resilience_LRTA}. Also consider $\Ef=\{\ef\}$ and $\Rcal_I = \{(q_4,q_5)\}$.
  The LRTA $\Acal_{\ref{couRTA5}\nf}$ is shown in \autoref{fig6_Resilience_LRTA},
  part of the concurrent composition $\CC(\Acal_{\ref{couRTA5}},\Acal_{\ref{couRTA5}\nf\obs})$ is shown in \autoref{fig7_Resilience_LRTA}.
  In $\CC(\Acal_{\ref{couRTA5}},\Acal_{\ref{couRTA5}\nf\obs})$, one can find the unique faulty transition
  $(q_0,\{q_0\}) \to (q_4,\{q_5,q_6\})$, where $(q_4,q_5)\in \Rcal_I$. Whenever the faulty event $\ef$ occurs,
  and subsequently whenever $\Acal_{\ref{couRTA5}}$ transitions to state $q_4$ via run
  $q_0 \xrightarrow[]{\ef/1} q_1 \xrightarrow[]{a/t} q_4 $, where $t$ must be in $[0,1]$,
  $\Acal_{\ref{couRTA5}}$ can return to normalcy via run $q_0 \xrightarrow[]{u/t} q_2 \xrightarrow[]{a/1} q_5$,
  and the occurrence of $\ef$ cannot be released because both runs produce timed label sequence $(a,1+t)$.
  There is no run which corresponds to \eqref{eqn4_Resilience_LRTA}. Then by \autoref{thm3_Resilience_LRTA}, 
  $\Acal_{\ref{couRTA5}}$ is resilient with respect to $\Ef$ and $\Rcal_I$.

 \begin{figure}[!htbp]
	\centering
	\begin{tikzpicture}
	[>=stealth',shorten >=1pt,thick,auto,node distance=2.5 cm, scale = 0.7, transform shape,
	->,>=stealth,inner sep=2pt]

	\tikzstyle{emptynode}=[inner sep=0,outer sep=0]

	\node[initial, state, initial where = left] (q0) {$q_0$};
	\node[state] (q2) [right of = q0] {$q_2$};
	\node[state] (q1) [above of = q2] {$q_1$};
	\node[state] (q3) [below of = q2] {$q_3$};
	\node[state] (q5) [right of = q2] {$q_5$};
	\node[state] (q4) [above of = q5] {$q_4$};
	\node[state] (q6) [below of = q5] {$q_6$};
	\node[state] (q7) [right of = q4] {$q_7$};
	\node[state] (q8) [right of = q5] {$q_8$};

	\path [->]
	(q0) edge node [above, sloped] {$\ef/[1,1]$} (q1)
	(q0) edge node [above, sloped] {$u/[0,1]$} (q2)
	(q0) edge node [above, sloped] {$u/[0,2]$} (q3)
	(q1) edge node [above, sloped] {$a/[0,1]$} (q4)
	(q2) edge node [above, sloped] {$a/[1,1]$} (q5)
	(q3) edge node [above, sloped] {$a/[0,0]$} (q6)
	(q4) edge node [above, sloped] {$u/[0,1]$} (q7)
	(q5) edge node [above, sloped] {$u/[0,1]$} (q8)
	(q7) edge [loop right] node [right] {$a/[0,1]$} (q7)
	(q8) edge [loop right] node [right] {$a/[0,1]$} (q8)
	(q8) edge [loop above] node {$u/[1,1]$} (q8)
	;

	 \end{tikzpicture}
	 \caption{An LRTA $\Acal_{\ref{couRTA5}}$, $\ell(u)=\ell(\ef)=\ep$, $\ell(a)=a$.}
	 \label{fig5_Resilience_LRTA}
   \end{figure}
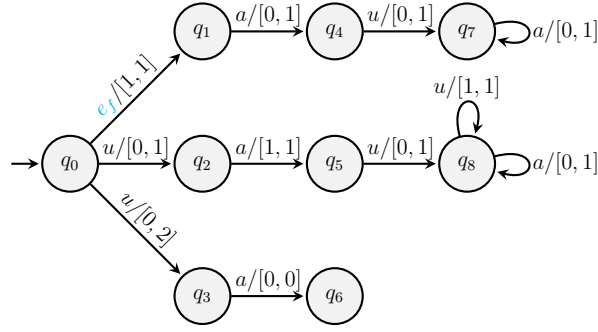

   \begin{figure}[!htbp]
	 \centering
	\begin{tikzpicture}
	[>=stealth',shorten >=1pt,thick,auto,node distance=2.5 cm, scale = 0.7, transform shape,
	->,>=stealth,inner sep=2pt]

	\tikzstyle{emptynode}=[inner sep=0,outer sep=0]

	\node[initial, state, initial where = left] (q0) {$q_0$};
	\node[state] (q2) [right of = q0] {$q_2$};
	\node[state] (q3) [below of = q2] {$q_3$};
	\node[state] (q5) [right of = q2] {$q_5$};
	\node[state] (q6) [below of = q5] {$q_6$};
	\node[state] (q8) [right of = q5] {$q_8$};

	\path [->]
	(q0) edge node [above, sloped] {$u/[0,1]$} (q2)
	(q0) edge node [above, sloped] {$u/[0,2]$} (q3)
	(q2) edge node [above, sloped] {$a/[1,1]$} (q5)
	(q3) edge node [above, sloped] {$a/[0,0]$} (q6)
	(q5) edge node [above, sloped] {$u/[0,1]$} (q8)
	(q8) edge [loop right] node [right] {$a/[0,1]$} (q8)
	(q8) edge [loop above] node {$u/[1,1]$} (q8)
	;

	 \end{tikzpicture}
   \caption{The LRTA $\Acal_{\ref{couRTA5}\nf}$, $\ell(u)=\ell(\ef)=\ep$, $\ell(a)=a$.}
	 \label{fig6_Resilience_LRTA}
	\end{figure}

	\begin{figure}[!htbp]
		\centering
	\begin{tikzpicture}
	[>=stealth',shorten >=1pt,thick,auto,node distance=3.5 cm, scale = 0.7, transform shape,
	->,>=stealth,inner sep=2pt]

	\tikzstyle{emptynode}=[inner sep=0,outer sep=0]

	\node[initial, elliptic state, initial where = left] (00) {$(q_0,\{q_0\})$};
	\node[elliptic state] (456) [right = 1cm of 00] {$(q_4,\{q_5,q_6\})$};
	\node[elliptic state] (78) [right of = 456] {$(q_7,\{q_8\})$};

	\path [->]
	(00) edge (456)
	(456) edge (78)
	(78) edge [loop right] (78)
	;

	 \end{tikzpicture}
	 \caption{A part of a concurrent composition $\CC(\Acal_{\ref{couRTA5}},\Acal_{\ref{couRTA5}\nf\obs})$.}
	 \label{fig7_Resilience_LRTA}
	\end{figure}
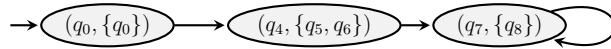

\end{example}

\section{Conclusion}
\label{sec:conc}

In this paper, we defined a notion of resilience for a labeled real-time automaton,
and designed an algorithm for verifying the resilience in $\mathsf{N}$2$\EXPTIME$ 
based on the basic tools --- concurrent composition \cite{Yan2024DiagnosisReal-TimeAutomata} and observer \cite{Zhang2024OpacityReal-TimeAutomata}. Related future topics include exploring other interesting notions of resilience, applications to practical scenarios 
such as how the notions of resilience describe the behavior of intelligent agents, etc.


\end{document}